\documentclass{article}
\usepackage{ijcai26}
\newcommand{\onepdfbuild}{}
\usepackage{times}
\usepackage{url}
\usepackage[hidelinks]{hyperref}
\usepackage[utf8]{inputenc}
\usepackage[small]{caption}
\usepackage{graphicx}
\usepackage{amsmath}
\usepackage{amssymb}
\usepackage{amsthm}
\usepackage{booktabs}
\usepackage{array}
\usepackage{tikz}
\usetikzlibrary{arrows.meta}
\newcolumntype{L}[1]{>{\raggedright\arraybackslash}p{#1}}
\newcommand{\E}{\mathbb{E}}

\theoremstyle{plain}
\newtheorem{theorem}{Theorem}[section]
\newtheorem{proposition}[theorem]{Proposition}
\newtheorem{lemma}[theorem]{Lemma}

\theoremstyle{definition}

\theoremstyle{remark}

\usepackage{xr-hyper}

\ifdefined\onepdfbuild
\newcommand{\Supp}{Appendix}\newcommand{\Supps}{Appendices}
\newcommand{\mainpapers}{}\newcommand{\onetwo}[2]{#1}
\else
\newcommand{\Supp}{Supplement}\newcommand{\Supps}{Supplements}
\newcommand{\mainpapers}{the main paper's }\newcommand{\onetwo}[2]{#2}
\fi

\newif\ifanon
\anonfalse

\title{First Attack, Final Offensive:\\ The Dark Forest on an Open Roster}

\author{
    Harvey Dam
    \affiliations
    Kahlert School of Computing, University of Utah
    \emails
    harvey.dam@utah.edu
}

\begin{document}

\maketitle

\begin{abstract}
The Dark Forest argument holds that a civilization that detects another should strike it at once. The choice it addresses is general: an agent that has located another nearby must decide whether to strike it, collaborate with it, or do neither, without knowing who else is there. Of sixteen treatments of this decision that we examined, ten assume a known set of players; the other six add unobserved agents only in games in which strikes and replies arrive instantly, in simulations, or in informal arguments. We show that when the roster is open, so that no agent knows even its list of nearby agents to be complete, the course that maximizes an agent's expected survival time if no unobserved agent acts still maximizes it, whatever else the agent believes, if and only if the agent gives at most a threshold probability to an unobserved agent acting, by attacking or helping the agent or by striking the located agent. Without further assumptions the threshold is below one half whatever the course, because an unobserved agent may reply to one course with help and to another with an attack, as soon as an agent as near as the located one could reply. If the agent knows the located agent completely and assumes that it would receive no help from unobserved agents, the threshold for collaboration can approach one. Even below the threshold the agent cannot know that the course it favors is in fact better: however much it knows of the one it has located, some world consistent with all it has observed makes another course better. These results do not apply to an agent that knows all but one of its courses would destroy it. For every other agent, no case for striking first, or for any other course, that assumes nothing about the unobserved agents withstands an even chance that one of them acts: some belief that gives that chance favors another course.

\end{abstract}

\section{Introduction}
\label{sec:intro}

An agent has located another nearby: it believes an attack of its own might do that agent damage. It can strike the other now, collaborate with it, or do neither. Which should it do? We call this the strike question. It arises wherever agents can harm one another without being sure who else is there, as on open networks of autonomous agents \cite{hewitt1991open,kouvaros2019open}, and has its best-known form among civilizations that detect one another across interstellar distances, the Dark Forest \cite{liu2008dark}.

\begin{figure}[t]
\centering
\begin{tikzpicture}[x=1cm,y=1cm,
  seen/.style={circle,draw,thick,minimum size=5.4mm,inner sep=0pt,font=\footnotesize},
  unseen/.style={circle,draw,thick,dashed,minimum size=5.4mm,inner sep=0pt,font=\footnotesize},
  attack/.style={-{Stealth[length=2.2mm]},line width=1.2pt},
  strike/.style={-{Stealth[length=2mm]},line width=0.9pt,double},
  help/.style={-{Stealth[length=2mm]},line width=0.9pt,densely dotted},
  lab/.style={font=\footnotesize}]
\node[lab] at (1.9,5.8) {$V$ does neither};
\node[lab] at (6.1,5.8) {$V$ strikes $U$};
\draw[gray] (4.0,3.75) -- (4.0,5.1);
\draw[gray] (4.0,1.75) -- (4.0,3.1);
\draw[gray] (0.0,3.6) -- (8.0,3.6);
\draw[gray] (0.0,1.6) -- (8.0,1.6);
\node[lab,anchor=west] at (0.0,5.35) {\emph{No unobserved agent acts: striking is better, by 0.722}};
\node[seen] (u1) at (1.2,4.65) {$U$}; \node[seen] (v1) at (2.6,4.65) {$V$};
\draw[attack] (u1) -- (v1);
\node[lab] at (1.9,3.95) {1.91 steps};
\node[seen] (u2) at (5.4,4.65) {$U$}; \node[seen] (v2) at (6.8,4.65) {$V$};
\draw[attack] (u2) to[bend left=24] (v2); \draw[strike] (v2) to[bend left=24] (u2);
\node[lab] at (6.1,3.95) {\textbf{2.63 steps}};
\node[lab,anchor=west] at (0.0,3.35) {\emph{An unobserved agent acts: striking is worse, by nearly 0.90}};
\node[seen] (u3) at (0.5,2.5) {$U$}; \node[seen] (v3) at (1.9,2.5) {$V$}; \node[unseen] (w3) at (3.3,2.5) {$W$};
\draw[attack] (w3) to[bend right=32] (u3); \draw[help] (w3) -- (v3);
\node[lab] at (1.9,1.9) {\textbf{nearly 2.80 steps}};
\node[seen] (u4) at (4.7,2.5) {$U$}; \node[seen] (v4) at (6.1,2.5) {$V$}; \node[unseen] (w4) at (7.5,2.5) {$W$};
\draw[attack] (w4) to[bend right=32] (u4); \draw[strike] (v4) -- (u4); \draw[attack] (w4) -- (v4);
\node[lab] at (6.1,1.9) {nearly 1.90 steps};
\node[lab,anchor=west] at (0.0,1.35) {\emph{$V$'s probability that an unobserved agent acts}};
\draw[line width=3.2pt] (0.4,0.45) -- (3.6,0.45);
\draw[line width=3.2pt,gray!45] (3.6,0.45) -- (7.6,0.45);
\draw (0.4,0.3) -- (0.4,0.6); \draw (7.6,0.3) -- (7.6,0.6);
\draw (3.6,0.3) -- (3.6,0.6); \draw (4.0,0.35) -- (4.0,0.55);
\node[lab] at (0.4,0.08) {$0$}; \node[lab] at (7.6,0.08) {$1$};
\node[lab,anchor=east] at (3.85,0.08) {$0.445$}; \node[lab,anchor=west] at (3.9,0.08) {$\frac{1}{2}$};
\node[lab] at (2.1,0.82) {striking is the better choice};
\node[lab] at (5.8,0.82) {not under every belief};
\end{tikzpicture}
\caption{How far the closed-roster answer, the course that is better if no unobserved agent acts, carries over, for a $U$ known to attack $V$ in every later step (two later steps; chances as in \S\ref{sec:open}). The dashed agent is one $V$ has not observed; arrows are attacks, the double arrow $V$'s strike; numbers are the expected number of steps, of three, at whose start $V$ is there. In the second row $W$ ends $U$ in the first step, before $U$ attacks, and, as soon as a near agent could reply, helps $V$ if $V$ has done neither (dotted arrow) and attacks $V$ if $V$ has struck, each with a chance near one. Both rows fit all $V$ has observed, so it cannot know that striking is better (Lemma~\ref{lem:open}(a)). Striking remains its better choice in expectation under every belief that gives probability at most $0.445$, the ratio of $0.722$ to $0.722+0.9$, to an unobserved agent acting, and not under every belief that gives more; for every course the threshold is below one half (Lemma~\ref{lem:open}(b), (c)).}
\label{fig:worlds}
\end{figure}

The answers given so far disagree. Jebari and Olsson-Yaouzis carry the two-player game of the Cold War, whose risk-dominant equilibrium is war, over to a first contact, where the lack of communication ``may lead players into a Hobbesian Trap'' \shortcite{jebari2018game}; Jebari and Asker argue that observing one nearby civilization would suggest many unobserved ones and a mechanism that keeps them from attacking, so that the trap ``would be avoided if both parties assumed the other is not likely to attack'' \shortcite{jebari2024dark}. Of the sixteen treatments we examined, ten fix the set of players, and there the facts about the players decide the answer (\S\ref{sec:closure}, \Supp~\ref{sec:supp-excluded}). Our premise is that no agent knows that set: it cannot know its list of others to be complete, since what it has not detected may still be there; and whoever is there may learn of what it does and reply, after a delay set by distance. We call a roster \emph{open} when no agent knows it to be complete, and give the question a model whose every premise is stated (\S\ref{sec:model}).

\textbf{We show that on an open roster an agent that seeks survival cannot know, when it decides, that the course that is better if no unobserved agent acts is in fact better, however much it knows of the agent it has located; and that this course remains its better choice in expectation, whatever else it believes, exactly while it gives at most a threshold probability to an unobserved agent acting, a threshold that is below one half for every course} (Lemma~\ref{lem:open}). Figure~\ref{fig:worlds} shows both. The contributions are these.
\begin{itemize}
\item \emph{With nothing assumed of the agents not observed} (Lemma~\ref{lem:open}), for a located agent near enough to reply in the next step. An agent is never certain that it has missed no one nearby, so over two or more steps after the first it can know of a course only that it is at least as good as a course that it has observed would end it: it cannot know even that striking an agent known to attack it is better than doing neither. Between two courses, neither of which it has observed would end it, the course that is better on a closed roster has the larger expected time under every belief that gives at most a threshold probability to an agent it has not observed attacking it, striking the one it has located, or helping it; under some belief that gives any larger probability, that course has the smaller. The threshold is below one half for every course, whatever the agent knows of the one it has located, and for a strike or doing neither it tends to zero as the horizon grows.
\item \emph{What restrictions restore} (Theorem~\ref{thm:main}). If the agents that an agent has not observed attack it or the located agent whatever it does, and neither reply to what it does nor help it, the agent can know one course to be at least as good as another exactly when it is so over the horizon and every shorter one, under each way they may end the located agent. If they also reply, and a course may turn them away from the agent as well as provoke them, then, supposing that none of them is near enough to reply within some number of steps and that none helps, the same holds over a horizon of at most that many steps; beyond it nothing more can be known than with nothing assumed. Where no such agent would help and the located agent is known completely, the threshold is at least as high as with nothing assumed: still below one half for a strike or doing neither, and for a collaboration it can approach one.
\item \emph{Displacement} (Theorem~\ref{thm:displace}). What can be known does not change when each agent not observed is required to be a hostile (an agent that values the capacity it removes) with one attack, which acts as is best for it, whose chance of success does not depend on the deciding agent's course, and which the deciding agent cannot harm. A strike that reveals its author to a weak attacker can then lead a strong one, less likely to find its target still there, to attack elsewhere, and so leave its author more time than doing neither.
\item \emph{Where existing treatments stand} (\S\ref{sec:closure}). The sixteen fall into three classes by what they let the agents one has not observed do; ten are wholly in the first, the closed roster, where their theorems hold.
\end{itemize}
\textbf{An answer that rests on the located agent alone, with no assumption about the unobserved agents, is therefore an answer under a bound on the probability that one of them acts: unless every other course is known to end the agent, a case for striking first, or for any other course, does not withstand an even chance that one of them acts, since some belief that gives that chance favors another course.}

\section{The model}
\label{sec:model}

We state the model as a list of premises, so that a claim that the premises do not determine something has a definite meaning. Nothing beyond the list is assumed.

\paragraph{In standard terms}
A world of the model is a finite extensive game with chance moves (those constructed in \S\ref{sec:result} fix the others' conduct and leave the deciding agent a decision problem): a record is a history, a course an action or a set of actions at a history, a conduct a behavior strategy, a joint outcome a sequential equilibrium of the game with control errors of a fixed size, that is, an equilibrium of Selten's perturbed game \cite{selten1975reexamination} with the corresponding minimum probabilities, and ``can determine'' knowledge.

\paragraph{Capacity and attack}
\emph{Capacity} is the potential to remove capacity, from now on and forever: whatever could be used to remove capacity is capacity, a weapon, a warning that brings someone else's strike, a subversion. An \emph{attack} is an attempt to remove capacity. What an attack removes is uncertain, and it need not be related to what others can do to the attacker. No attack is certain to remove anything, and an attack that is certain to remove nothing is degenerate. An agent is \emph{ended} when its capacity is zero, when it can no longer make decisions or its decisions no longer matter; a strike ends its target exactly when it removes all its capacity. An agent can also be ended without having been attacked.

\paragraph{Locating and the roster}
To \emph{locate} an agent is to believe that an attack of one's own might do it damage; locating is an update of information. An attack can be aimed only at an agent its author has located. There are finitely many agents, and no agent knows its list of the others to be complete: the \emph{roster is open}. Nothing in the model makes one agent unable to locate another: a world of the model may not restrict by fiat whom an agent can locate or strike.

\paragraph{Aims}
A \emph{survival-seeker} wants to keep decision-making ability as long as possible, and nothing else is assumed of it, since survival is the only assumption one can make about agents that have survived. Its criterion is the expected time for which it keeps decision-making ability, over a horizon that may end at a time it does not foresee. A \emph{hostile} values the capacity it expects to remove from now until it loses decision-making ability; it does not value its survival as such, and it need not have attacked anyone. Whether an agent is a survival-seeker or a hostile is private.

\paragraph{Acting}
At every moment every agent chooses an action, and ``do nothing'' is one of them. Against a given located agent the actions divide into a strike set and a not-strike set, each with possibly many members, and we speak of three \emph{courses}: to strike the located agent, to collaborate with it, or to do neither. To reply to a strike is to act on information that includes it; a reply need not be anything in particular. An agent may find another's survival valuable. One that does may act to lengthen it, whether or not the other has observed it; and a witness that does, and is more capable than the agent attacked, has a reason to act against the attacker though it was not attacked itself.

\paragraph{Probabilities, options, determining}
Agents share a prior and each conditions on its own information; no independence is assumed, so no dependence is excluded. An agent has the option of its best action and nothing more: it is more likely to take an action the more that action favors what it wants, and nothing says it always takes its best action. An agent's conduct at a record may be a probability distribution over its actions. An agent \emph{can determine} a statement when the statement is true in every state of affairs consistent with what the agent has observed: knowledge in the usual sense, truth throughout the agent's information cell \cite{aumann1976agreeing,fagin1995reasoning}. An agent's expectation of each course under the shared prior is fixed by what it has observed; what may or may not be determined is which course is best in the state of affairs that obtains, which the prior is a belief about and not knowledge of. A \emph{joint outcome} is a conduct for every agent at every record, best for that agent at that record given its belief, the shared prior and the conduct at every other record, its own included; every record has positive probability, each chosen action being carried out as chosen with probability $1-\varepsilon$ and as each other available action with positive probability, so that ``best given its information'' has meaning at every record. Nothing in the model says that agents behave so; statements about joint outcomes are conditional.

\paragraph{Space, detection and delay}
Three further premises. First, capacity occupies space. Knowledge must be stored somewhere, storage takes physical space, and the agent that would use it takes space; so an agent's capacity has a physical extent whatever its form, and it is not absolutely cold. Capacity is therefore detectable, which does not mean that it is detected. Light speed limits the region within which another agent is detectable, and in no part of that region of nonzero extent is an agent certain that it has observed every agent of significant capacity: that it has missed one nearby can be extremely unlikely, the more so with what reliable others have observed and passed on, and is not impossible. An attack already made is detectable and need not have been detected: an agent that could detect some effect of it has not thereby observed it. Second, nothing travels faster than light: an agent learns of another's acts, and anything it does in reply takes effect, after a delay that grows with their distance, for located and unobserved agents alike. The third is about discovery: given its record of what it has found so far, an agent's prior over the agents it has yet to discover gives positive probability to there being more, whatever has been found, and gives the capacity it will discover within any span of steps a finite expectation. Priors so constrained, giving some weight to an unobserved agent within any distance, letting each agent reply only after its delay, never closing the roster and keeping what remains to be found finite in expectation, we call \emph{plausible}.

\section{The main result}
\label{sec:result}

\subsection{A class of worlds}
\label{sec:class}

``Cannot determine'' is an existence statement: two states of affairs consistent with what the agent has observed that give opposite answers. We construct a class of worlds and ask what is determinable in it; the negative statements then hold in the model whatever else it allows, since more states only make determining harder; the positive ones are statements about the class. Six rules define the class, which describes the agent's situation at the moment it decides, with a first step and $n$ later steps; three more, (s5) to (s7), are in \Supp~\ref{sec:supp-after} with Theorem~\ref{thm:ladder}, on what the agent can determine as a hostile, in hindsight, and by pooling records.

\begin{enumerate}
\item[(s1)] The agent $V$ takes one of three courses and keeps to it: strike the located agent $U$ in every step while $U$ is there; collaborate with $U$; do neither.
\item[(s2)] What $V$ has observed when it decides is a record, which may be anything, together with, for striking and for collaborating, whether it has observed that the course would itself leave it with no capacity. A course so observed ends $V$ when it is taken.
\item[(s3)] What the others do in each later step may depend on the course $V$ took and on the step. $U$, while it is there, attacks $V$ in that step or does not. An agent that $V$ has not observed attacks $V$ in that step or does not, and strikes $U$ in that step or does not. If $V$ collaborates, $U$ collaborates in return or does not. Nothing in what $V$ has observed restricts what the agents it has not observed do.
\item[(s4)] Chances, the same in every world of the class, so that $V$ may know them: $V$ gets through a step in which no one attacks it with probability $a$, or $a^{+}\ge a$ in a later step while a $U$ that collaborates in return is there; an attack by $U$ ends it with probability $h$; $V$'s strike ends $U$ with probability $s$; apart from that $U$ gets through a step with probability $e$; an attack by an agent not observed, on $V$ or on $U$, ends its target with a probability that belongs to that agent, the same after every course for an attack made regardless of the course and free to differ from course to course for a reply. All lie strictly between $0$ and $1$, except that $a^{+}$ and $e$ may be $1$. An attack that does not end its target leaves it as it was, and what is done in a step takes effect at the next: an agent that is there at the start of a step acts in that step even if it is ended in it. No one attacks $V$ in the first step, where an attack made whatever its course would weigh on every course alike.
\item[(s8)] Distance and delay. Every agent other than $V$ is at a distance from it, and delay is by distance alone: an agent's conduct in a later step may depend on $V$'s course only from the step at which the course can have reached it, and an attack it makes in reply lands only from the step at which a strike can have traveled back. Agents near enough that a reply lands in the next step are \emph{near}, the step being that unit; a located agent farther than that replies by its distance. Of the agents $V$ has not observed, what one does regardless of $V$'s course can land in any step, an attack already made being perhaps unobserved; what one does in reply to $V$'s course lands no sooner than later step $D\ge1$, the round trip to the nearest of them counted in steps. $V$ is certain of no $D$ beyond $1$, at which a reply lands as soon as a near agent's (\S\ref{sec:model}); a larger $D$ is a supposition, that $V$ has missed no one nearer.
\item[(s9)] An agent that $V$ has not observed helps $V$ in a later step or does not. Help raises the chance with which $V$ gets through that step apart from any attack, $a$ or $a^{+}$ in (s4), to a probability that belongs to the helper and may be anything up to one. Like an attack, help is given regardless of $V$'s course, in any later step, or in reply to the course, no sooner than step $D$.
\end{enumerate}

Rule (s3) is the open roster: an agent $V$ has not observed may attack it after a strike, after a collaboration, or after neither, and nothing $V$ has observed fixes which; it may also strike $U$, a world being forbidden to restrict whom an agent can strike. That $V$ has not observed an agent bounds neither what that agent's attack can remove nor how near it is. Rule (s9) is the premise of \S\ref{sec:model} that an agent may want another's survival. Rule (s1) is a stipulation of the class: when $V$ chooses again at every step, and locates and may strike the agents it has not observed, each course is still best in some world consistent with its record (\Supp~\ref{sec:supp-worlds}). The value of a course is the expected number of steps at whose start $V$ is there, computed from the rules.

Rule (s8) limits how far ahead $V$ can determine anything. What an agent not observed does regardless of $V$'s course can land at any time: its attack on $V$ is a weight on $V$'s chance of being there at each step, common to every course, no greater than one and not increasing; its strike may end $U$ in any step, the first included, and then before $U$ has done anything. What it does in reply to $V$'s course lands no sooner than step $D$: a second weight, equal to one through step $D$ and depending on the course. Write $p^x_t(j)$ for $V$'s chance of being there at the start of step $t$ after course $x$ when no agent not observed attacks or helps it and a strike of one of them ends $U$ in step $j$, $j=0$ being the first step, and $p^x_t(\infty)$ when none does; $x$ \emph{dominates} $y$ at $j$ over a set of horizons when, at each horizon in the set, the partial sum of $p^x(j)-p^y(j)$ up to that horizon is non-negative. Below, $j$ ranges over the steps and $\infty$.

An agent $V$ has not observed \emph{acts} in a world when it attacks $V$, strikes $U$ or helps $V$ in a step before step $n$, the first included, on the branch of either course compared. A \emph{belief} is a prior over the worlds of the class, conditioned on what $V$ has observed, which may leave it unsure of $U$; the model fixes no prior, and the results range over beliefs. Write $v_k(x;j)=\sum_{t=0}^{k}p^x_t(j)$ and $v=v_n$.

\subsection{Nothing assumed of the agents not observed}
\label{sec:open}

The lemma takes the class as it stands, with no supposition about the round trip ($D=1$). It says what $V$ can know, and how much probability $V$ may give to the agents it has not observed before the closed-roster answer, the course that is better if none of them acts, stops being its better choice in expectation.

\begin{lemma}
\label{lem:open}
Let a survival-seeker $V$ have located a near agent $U$, with a horizon of $n\ge2$ later steps and $D=1$, and let what $V$ knows of $U$, its aim, its conduct after each course and the chances of (s4) that concern them, leave a set of possibilities. Let $x$ and $y$ be two courses.
\begin{enumerate}
\item[(a)] $V$ can determine that $x$ gives it at least as much time as $y$ exactly when it has observed that $y$ would end it.
\item[(b)] Let $V$ have observed of neither course that it would end it, and let $m>0$ be the amount by which $V$ expects $x$ to give more than $y$ over the worlds in which no agent it has not observed acts, the closed-roster margin: $m=v(x;\infty)-v(y;\infty)$ when $V$ knows $U$ completely. Then $x$ gives $V$ the larger expected time under every belief with that margin that gives probability at most $w$ to the worlds in which an agent it has not observed acts, exactly when
\[
w\ \le\ \frac{m}{m+a(n-1)}.
\]
\item[(c)] $m<a(n-1)$, so the threshold of (b) is below one half for every course. If $x$ is a strike or doing neither, $m<S=a^{2}+\dots+a^{n}$ and the threshold is below $S/(S+a(n-1))$, which falls to zero as the horizon grows.
\end{enumerate}
\end{lemma}

Over one later step every course not observed to end $V$ gives the same.

\begin{proof}
(a) A course observed to end $V$ gives it the first step alone, and no course gives less. If $y$ is not so observed, let agents $V$ has not observed end $U$ in the first step with a chance near one, so that any course not observed to end $V$ leaves it there at step $t$ with probability near $a^{t}$, and let another attack $V$ from the first later step on after $x$ and not after $y$: then $y$ gives more. (b) In every world $x$ gives more than $1+a$: no one attacks $V$ in the first step, what is done in the first later step takes effect at the second, and no attack is certain. In every world $y$ gives at most $1+a+a(n-1)$, $V$ being there at the start of each step after the first later one with probability at most $a$. So $y$ gives less than $a(n-1)$ more than $x$ in every world, and under a belief that gives probability $w'$ to the worlds in which an agent not observed acts the expected advantage of $x$ is $m$ when $w'=0$ and more than $(1-w')m-w'a(n-1)$ otherwise: it is positive for $w'\le m/(m+a(n-1))$. Conversely, let agents $V$ has not observed strike $U$ in the first step, whatever $V$ does, help $V$ in every later step after $y$ and not after $x$, and attack $V$ in the first later step after $x$ and not after $y$, each with a chance as near one as one likes. Whatever $U$ is, $y$ then gives nearly $1+a+a(n-1)$ and $x$ nearly $1+a$. When $w>m/(m+a(n-1))$, a belief that gives this world a probability above that threshold, at most $w$ and short of one, and the rest to the worlds in which no one acts, with $V$'s belief about $U$ the same in both, gives $y$ the larger expected time. (c) When no such agent acts every course gives at most $1+a+a(n-1)$, and at most $1+a+S$ if it is a strike or doing neither, since these leave $V$'s chance of getting through each step at most $a$; and $y$ gives more than $1+a$.
\end{proof}

Take a known aggressor, a $U$ that attacks $V$ in every later step, with $a=9/10$, $h=99/100$, $s=9/10$ and $e=1$. Over two later steps striking gives $2.63$ expected steps against $1.91$ for doing neither when no agent not observed acts, so $m=0.722$, and $a(n-1)=0.9$: striking is the better choice under every belief that gives probability at most $8019/18019\approx0.445$ to an agent not observed acting (Figure~\ref{fig:worlds}). The threshold is $0.398$ over six later steps and $0.039$ over two hundred. For a partner that returns a collaboration and otherwise does nothing ($a^{+}=1$), collaborating against doing neither, the threshold is $0.459$ over sixty later steps. The arithmetic of (b) is that of a prior contaminated with weight $w$ by an arbitrary one (\S\ref{sec:related}). Help to $U$, which the class leaves out, changes nothing in (b), and the bound of one half holds when $V$ chooses again at every step (\Supp~\ref{sec:supp-exante}).

\subsection{What restrictions restore}
\label{sec:restrict}

Each part of the theorem below restricts the agents $V$ has not observed in a stated way and says what can then be known or chosen; none of the restrictions is a premise of the model.

\begin{theorem}
\label{thm:main}
Let a survival-seeker $V$ have located a near agent $U$, with a horizon of $n$ later steps, $D$ the round trip of (s8), and let what $V$ knows of $U$, its aim, its conduct after each course and the chances of (s4) that concern them, leave a set of possibilities, each of which fixes the $p^x_t(j)$. In (a) to (d), $V$ can determine that a course $x$ gives it at least as much time as another course $y$ exactly when the stated condition holds for every possibility.
\begin{enumerate}
\item[(a)] If nothing $V$ has not observed acts on it or on $U$ (the closed roster): $x$ gives at least as much as $y$ over the horizon, so what $V$ knows of $U$ is all that decides.
\item[(b)] If the agents $V$ has not observed do not help it and their conduct does not depend on its course: for every $j$, $x$ dominates $y$ at $j$ over every horizon up to the one at hand, which is time dominance \cite{ekern1981time,foster2003ranking} under each way $U$ may be ended.
\item[(c)] If they do not help it and can be provoked by a course but never turned away by one, that is, reply only by attacks on $V$ that make every later step at least as dangerous to it after a strike or a collaboration as after doing neither: for $x$ doing neither, the condition of (b); for $x$ a strike or a collaboration, the condition of (d).
\item[(d)] If they may also be turned away by a course, being less dangerous to $V$ after a strike or a collaboration than after doing neither, as the model allows and Theorem~\ref{thm:displace} exhibits, and, by supposition, none of them is nearer than the round trip $D$: over a horizon $n\le D$, and if they do not help $V$, the condition of (b); over a longer one, as in Lemma~\ref{lem:open}(a), exactly when $V$ has observed that $y$ would end it.
\end{enumerate}
In (e) and (f), $V$ has observed of neither course that it would end it, $n\ge2$, and $m>0$ is the closed-roster margin of Lemma~\ref{lem:open}(b).
\begin{enumerate}
\item[(e)] A bound stronger than that of Lemma~\ref{lem:open}(b) holds for the worlds in which no agent $V$ has not observed helps it. Let $U$ be known completely and $L=\max_j v(y;j)-(1+a)\le a(n-1)$. Then $x$ gives $V$ the larger expected time under every belief that gives probability at most $w_h$ to the worlds in which such an agent helps $V$ and at most $w_a$ to those in which one acts and none helps, exactly when $w_h\bigl(m+a(n-1)\bigr)+w_a(m+L)\le m$. With no help the threshold is $m/(m+L)$: still below one half when $x$ is a strike or doing neither, and it can be as near one as one likes when $x$ is a collaboration (Proposition~\ref{prop:exante}).
\item[(f)] If, by supposition, none of them is nearer than the round trip $D$, with $D<n$, so that no reply to $V$'s course, by attack, strike or help, lands before step $D$, and if $x$ gives at least as much as $y$ through step $D$ for every possibility of $U$ and under every conduct of the agents not observed that is the same after both courses: the statement of Lemma~\ref{lem:open}(b), read for the worlds with that round trip, holds with $a(n-D)$ in place of $a(n-1)$. Whatever the courses give through step $D$, after it $y$ can gain on $x$ by as near $a(n-D)$ as one likes.
\end{enumerate}
The model grants the supposition of (d) and (f) at $D=1$ only, where over two or more later steps they are Lemma~\ref{lem:open}(a) and (b). Where the condition of part (b) holds over the horizon, $x$ gives at least as much as $y$ in every world of the class in which $V$ has missed no agent near enough to reply within the horizon and none helps it.
\end{theorem}

\begin{proof}[Proof sketch]
(b) An attack on $V$ made whatever it does multiplies its chance of being there at step $t$ by a weight $C_t$, the same after every course and not increasing with $t$. By Abel summation $\sum_t C_t\bigl(p^x_t(j)-p^y_t(j)\bigr)\ge0$ for every such weight exactly when every partial sum of $p^x(j)-p^y(j)$ is non-negative, at each $j$ at which a strike made whatever $V$ does may end $U$. (d) Beyond the round trip, take the world of Lemma~\ref{lem:open}(a) with its attack made from step $D$ on. (e) Where such agents act and none helps, $y$ gives less than $L$ more than $x$, and that too is approached (Proposition~\ref{prop:exante}); the two bounds combine as in Lemma~\ref{lem:open}(b). (f) Through step $D$ the courses face the same conduct and $x$ is not behind; after it $y$ gains at most $a$ in each step, and the world of Lemma~\ref{lem:open}(b) with its help given whatever $V$ does and its attack landing in step $D$ approaches that. Parts (a) and (c), the horizons within the round trip and the variant with no strike on $U$ in the first step are in \Supps~\ref{sec:supp-worlds} and~\ref{sec:supp-moved}.
\end{proof}

With a round trip of $D=2$ supposed and no help, $V$ can determine that striking the known aggressor is at least as good as doing neither over two later steps, since striking is at least as good over one and over two later steps whenever an agent not observed ends $U$ and when none does; over a longer horizon it cannot, by (d). Where the agents not observed may also help $V$, every comparison that (b), (c) and the first half of (d) leave undetermined stays undetermined, and (b) needs a stronger condition (\Supp~\ref{sec:supp-exante}).

With no help, (e) raises the four thresholds of \S\ref{sec:open} to $0.471$, $0.473$, $0.474$ and $0.85$. By (f), the bound of one half is the bound for a reply that lands as soon as a near agent's, the one delay $V$ cannot exclude: for the aggressor over six later steps the threshold is $0.398$, $0.453$, $0.525$, $0.623$ and $0.768$ for $D=1,\dots,5$. If help can raise $V$'s chance of getting through a step only to $\beta<1$, $\beta\ge a^{+}$, the $a(n-1)$ of Lemma~\ref{lem:open}(b) becomes $a(\beta+\dots+\beta^{n-1})<a\beta/(1-\beta)$ (\Supp~\ref{sec:supp-exante}): the fall to zero in Lemma~\ref{lem:open}(c) rests on help that can make $V$ nearly certain to get through a step. Restricting beliefs to plausible priors (\S\ref{sec:model}) changes nothing in Lemma~\ref{lem:open}(b), (c) or part (e).

\subsection{Displacement}
\label{sec:displace}

\begin{figure}[t]
\centering
\begin{tikzpicture}[x=1cm,y=1cm,
  seen/.style={circle,draw,thick,minimum size=5.4mm,inner sep=0pt,font=\footnotesize},
  unseen/.style={circle,draw,thick,dashed,minimum size=5.4mm,inner sep=0pt,font=\footnotesize},
  strong/.style={-{Stealth[length=2.4mm]},line width=1.6pt},
  weak/.style={-{Stealth[length=2mm]},line width=0.5pt},
  strike/.style={-{Stealth[length=2mm]},line width=0.9pt,double},
  away/.style={-{Stealth[length=2mm]},line width=0.9pt,densely dotted},
  lab/.style={font=\footnotesize}]
\node[lab] at (1.9,2.8) {$V$ does neither};
\node[lab] at (6.1,2.8) {$V$ strikes $U$};
\draw[gray] (4.0,0.1) -- (4.0,2.5);
\node[seen] (u3) at (0.5,1.5) {$U$}; \node[seen] (v3) at (1.9,1.5) {$V$};
\node[unseen] (a3) at (3.3,2.05) {$A$}; \node[unseen] (b3) at (3.3,0.95) {$B$};
\draw[strong] (b3) -- (v3);
\node[lab] at (1.9,0.35) {1.97 steps};
\node[seen] (u4) at (4.6,1.5) {$U$}; \node[seen] (v4) at (5.9,1.5) {$V$};
\node[unseen] (a4) at (7.2,2.05) {$A$}; \node[unseen] (b4) at (7.2,0.95) {$B$};
\draw[strike] (v4) -- (u4); \draw[weak] (a4) -- (v4); \draw[away] (b4) -- (7.95,0.35);
\node[lab] at (5.9,0.35) {\textbf{4.87 steps}};
\end{tikzpicture}
\caption{Displacement (Theorem~\ref{thm:displace}). Dashed agents are ones $V$ has not observed. A strong hostile $B$ attacks $V$ (thick arrow) unless a strike (double arrow) reveals $V$ to a weak hostile $A$, which attacks first (thin arrow), so that $B$, now less likely to find $V$ still there, uses its one attack elsewhere (dotted). Numbers are the steps, of seven, at whose start $V$ expects to be there: the strike gives $V$ more time though $V$ threatens neither $A$ nor $B$.}
\label{fig:displace}
\end{figure}

That a course may turn the agents $V$ has not observed away from it is, in Theorem~\ref{thm:main}(d), an allowance of the model. It need not be taken as an assumption, and it need not be deterrence: an attacker that has reasons for what it does, one strength, and nothing to fear from $V$ can be displaced. Rule (R) restricts every agent $W$ that $V$ has not observed. $W$ is a hostile that values the capacity removed by its own attacks and has one attack to make, on the agent it has located, $V$ or $U$, or elsewhere, where it expects to remove $v_W>0$. Its attack ends a target that is there with a chance $\lambda_W$, its strength, that is the same whatever $V$'s course. A strike or a collaboration reveals $V$ to it at least as much as doing neither. Attacks that land on one target in the same step land in an order the world fixes, a later one finding its target only if it is still there. $W$ is ended without being attacked, a natural loss, with probability $\delta_W>0$ per step, and $V$ has no action against it; the theorem does not rest on that (\Supp~\ref{sec:supp-worlds}). Each of its choices comes out as chosen with probability $1-\varepsilon$ and as each of the other two with $\varepsilon/2$, $\varepsilon<1/2$, and the conducts of these agents form a joint outcome among them (\S\ref{sec:model}), $U$'s conduct and $V$'s course being given, in which every choice is strictly best at its record. Nothing $V$ has observed restricts $v_W$, $\lambda_W$, $\delta_W$, what $W$ has located, or where it is.

\begin{theorem}
\label{thm:displace}
Under rule (R), $V$ can determine, when it decides, exactly what Theorem~\ref{thm:main}(d) states, and whenever it cannot determine that a course gives it at least as much as another, some plausible prior gives it less in expectation. In particular, in some joint outcomes a strike, or a collaboration, gives $V$ more time than doing neither over every horizon beyond the round trip, though no agent's strength depends on $V$'s course and $V$ threatens none of them.
\end{theorem}

\begin{proof}
To an agent $W$ of the rule its attack on $V$ is worth $\lambda_W\pi_W$, where $\pi_W$ is the probability, on what $W$ knows, that it has located $V$ and that $V$ is there when the attack lands. $V$'s course enters $\pi_W$ in three ways and no other: by what it reveals, by what it leaves of $V$ by then, and by who else attacks $V$ first. The first two can only draw $W$ to $V$, after a course that reveals $V$ or keeps it there. Together with attacks made regardless of $V$'s course, they give every world that the proof of Theorem~\ref{thm:main} needs but one: the world in which an agent attacks $V$ after $V$ has done neither and not after a course that reveals $V$ (\Supp~\ref{sec:supp-closure}). The third way gives that one. Let $x$ be doing neither, $y$ a strike or a collaboration not observed to end $V$, and $n>D$. Let agents that strike $U$ whatever $V$ does end it in the first step with a probability as near one as one likes, so that, apart from the two agents that follow, $V$ is there at step $t$ with probability $a^{t}$ after either course. Agent $A$ locates $V$ from $y$ and not from $x$; it has a small strength $\kappa$ and expects so little elsewhere that attacking $V$ at once is strictly best for it; its attack lands in step $D$. Agent $B$ has located $V$ after both courses; it has strength $\lambda$ and $\delta_B\ge1/2$, its attack lands in step $D$ after $A$'s, and it expects $v_B$ elsewhere, with
\[
\lambda a^{D}\bigl(1-(1-\varepsilon)\kappa\bigr)<v_B<\lambda a^{D}.
\]
After $x$, $A$ cannot attack $V$, and $B$'s attack is worth $\lambda a^{D}$, more than $v_B$. After $y$, $A$'s attack comes first and $B$'s is worth the left side, less than $v_B$. Doing nothing for a step or more yields $B$ at most $1-\delta_B$ times the sum of the two values, whatever it learns meanwhile, which is less than the larger of the two. So $B$ attacks $V$ after $x$ and uses its attack elsewhere after $y$, each choice strictly best: $A$ has displaced it. At each step $t>D$, $V$ is then there with probability at most $a^{t}\bigl(1-(1-\varepsilon)\lambda\bigr)$ after $x$, and at least $a^{t}(1-\kappa)\bigl(1-\lambda\varepsilon/(2-\varepsilon)\bigr)$ after $y$, where $B$ attacks it only by mistake. For $\kappa$ small the second is the larger, since $(1-\varepsilon)(2-\varepsilon)>\varepsilon$. So $y$ gives $V$ more over every horizon beyond $D$, and $V$ cannot determine that doing neither gives at least as much. The rule in full, the other worlds and the prior are in \Supp~\ref{sec:supp-closure}.
\end{proof}

An instance: with $a=9/10$, $D=1$, six later steps, a $U$ that does not attack $V$, $\lambda=99/100$, $\kappa=1/10$, $\delta_B=1/2$ and $\varepsilon=1/100$, $B$'s attack on $V$ is worth $0.891$ to it after $V$ has done neither and $0.803$ after a strike, against $0.847$ elsewhere, and $V$ expects to be there at the start of $1.97$ of the seven steps after doing neither and of $4.87$ after the strike (Figure~\ref{fig:displace}).

The mechanism is known in kind. It is that by which a firm slightly ahead in a patent race preempts its rivals \cite{fudenberg1983preemption}, ecology's ``risk reduction caused by predator-predator interactions'' \cite{sih1998emergent}, and the reasoning of Konrad's second attacker, which raids the defender of second-highest value ``as the defender with the highest raiding value has already been attacked'' \cite{konrad2024collective}; where hunters share the benefit, free-riding gives the volunteer's dilemma \cite{diekmann1985volunteer}, which Navarro and Pereau apply to this strike \shortcite{navarro2026game}. Here each values only what it removes itself, and the target's own act draws the hunter that displaces the other. What the theorem adds to Theorem~\ref{thm:main}(d) is that its allowance is not an assumption and needs no threat: the strong attacker is displaced, not deterred.

\section{Where models live}
\label{sec:closure}

By what they allow the agents an agent has not observed to do, models of the strike question fall into three classes. Theorem~\ref{thm:main} says what can be determined in each when nothing one has observed restricts what those agents do within the class's allowance. Its conditions are exact in the worlds of \S\ref{sec:result}, and no more can be determined in a model that allows those worlds and others. In the first there are no such agents, and the facts about the agents one faces determine every comparison (part (a)). In the second they attack one or the agent one has located, but not in reply to what one does, and a comparison is determined by comparing the courses over the horizon and every shorter one, under each way that agent may be ended (part (b)). In the third they reply to what one does, each after the round trip its distance imposes: supposing no one missed near enough to reply within the horizon, a comparison is determined as in the second, and without the supposition none is, except against a course that one has observed would leave one with no capacity (part (d)). Part (b) is for agents that only attack; with helpers it needs a stronger condition (\S\ref{sec:restrict}).

Ten of the sixteen treatments of the strike question we examined are wholly in the first class, placed there by a premise given in \Supp~\ref{sec:supp-excluded}. Their theorems stand there, and the class is the right one wherever everyone able to act within the horizon has been observed. Three reach beyond it in a game or a simulation. Navarro and Pereau \shortcite{navarro2026game} have one game with anonymous others that act whatever one does, the second class, and games in which a civilization whose position is revealed is destroyed at once by one outside the game, the third class with no round trip. Ou \shortcite{ou2026darkforest}, whose theorem is in the first class, also proves that a striker's payoff falls with the number of third parties that may see its strike, with no time in it, and simulates strikes at light speed that expose the striker to others, the third class. Tan's agent-based model \cite{tan2026emergence} has civilizations that enter and leave and observers whose hostility an attack raises by a stipulated rule, the attacker reckoning with its target alone: a simulation in the third class. The other three argue in words: Korhonen \shortcite{korhonen2013mad} and Yu \shortcite{yu2015dark} that an attacker must reckon with witnesses it has not detected, the third class, and Jebari and Asker \shortcite{jebari2024dark} from an unobserved multitude kept from attacking by a mechanism they leave open, beside a two-player game of the first class.

The universe is in the third class, by three facts: no agent can know its list of others complete, nothing rules out that an agent not on the list observes an attack and replies to it, and nothing travels faster than light, so such a reply comes after a delay set by distance. We bound what can be determined in the third class, prove the bound exact in the worlds of \S\ref{sec:result}, and give the probability of such agents acting up to which the course better without them remains the better choice (Lemma~\ref{lem:open}). Of the universe we claim only its class.

\section{Related work}
\label{sec:related}

The Dark Forest gave the strike question its best-known form \cite{liu2008dark}. Section~\ref{sec:closure} places sixteen treatments of it; those on closed rosters include Jebari and Olsson-Yaouzis \shortcite{jebari2018game}, Kuusela and Roy \shortcite{kuusela2024higher} and Tan \shortcite{tan2026reconstructing}. Conflict theory's reciprocal fear of surprise attack \cite{schelling1960strategy}, security dilemma \cite{jervis1978cooperation}, bargaining and commitment failures \cite{fearon1995rationalist,powell2006commitment,krainin2017preventive}, arming types \cite{baliga2004arms,kydd1997sheep}, provocateurs and false flags \cite{baliga2012manipulating,baliga2020attribution}, truels \cite{kilgour1997truel}, defense of several targets \cite{konrad2024collective}, conflict before third parties \cite{chassang2010conflict,che2026first}, attribution games \cite{edwards2017strategic,welburn2023cyber} and the credibility of deterrent threats \cite{zagare2000perfect,press2005calculating} keep the roster known; here the agent that may reply may be one the deciding agent has not observed. Poisson and random-player games \cite{myerson1998population,milchtaich2004random} and open multiagent systems \cite{hewitt1991open,kouvaros2019open} model not knowing who is there, and games with unawareness \cite{heifetz2006interactive,halpern2008interactive} model agents that do not conceive of some players or actions. Here the agent conceives of others it has not observed, and our question is whether any facts about the agents present determine the answer. ``Can determine'' is unanimity over every prior on the states consistent with one's information, Bewley's criterion \cite{bewley2002knightian} with the largest set of priors; prior-free solution concepts choose by the worst case, for an unknown number of players \cite{ashlagi2006resource} and for uncertain payoffs \cite{aghassi2006robust}. Lemma~\ref{lem:open} lies between that unanimity and choice under one prior: its threshold is the largest weight of contamination by an arbitrary prior \cite{kopylov2009choice} under which the closed-roster answer stays unanimous, the counterpart of $1-p$ for a $p$-dominant equilibrium \cite{kajii1997robustness}.

\section{Conclusion}
\label{sec:conclusion}

We asked how far the answer that the facts about a located agent give to the strike question carries over when no agent knows its list of others to be complete. In the worlds of \S\ref{sec:result}, an agent that seeks survival and has located a near agent cannot know when it decides, over two or more later steps, that the answer is right, unless every alternative is a course that it has observed would end it. The answer remains its better choice in expectation, whatever else it believes, exactly while it gives at most a threshold probability to an unobserved agent attacking it, striking the agent it has located, or helping it; since such an agent may help as well as harm, and as soon as a near agent could reply, the threshold is below one half for every course.

The results recommend no course: advice to strike first, or to refrain, that is derived from the located agent alone carries a bound on the probability that an unobserved agent acts. Which agents would help another, and why, is a premise here and is not modeled; for harm Theorem~\ref{thm:displace} gives reasons.

\bibliographystyle{named}
\bibliography{references}

\appendix
\numberwithin{equation}{section}

\section{The class of worlds, in full}
\label{sec:supp-worlds}

Rules (s1) to (s4), (s8) and (s9) are stated in \onetwo{\S\ref{sec:class}}{the main paper (\S\ref{sec:class})}, and rules (s5) to (s7) in \Supp~\ref{sec:supp-after}. This section adds one more rule, (s7$'$); says which rules are premises of the model and which are stipulations of the class; gives the worlds used in the proofs of Theorems~\ref{thm:main} and~\ref{thm:ladder}; and gives the extensions of the class, each with an example in which its rules hold. Except where rule (s9) is named, the statements of this section are about the class without it, in which the agents $V$ has not observed act by attacks alone.

\paragraph{(s7$'$) Forests of any shape} Each agent has located one other agent, its target; agents need not be in pairs. An agent attacks only agents that have located it, and what it does toward one may depend on that agent's course toward it and on the step, and need not be what it does toward another that has located it. No agent strikes in the first step. In a step in which no one attacks it, every agent gets through with one and the same probability.

\paragraph{Which rules are the model's premises} That a course can leave its author with no capacity, and that such an agent is ended. The last sentence of (s3), that nothing an agent has observed fixes what the agents it has not observed do later, and for (s6) that nothing an agent has observed, at any time, fixes what another would have done had the first acted otherwise. The agent not observed that is provoked by a strike or a collaboration, or turned away by either. The collaboration returned that lengthens the agent's time. That no attack is certain to remove anything, that an attack may remove all its target has, that an agent can be ended without having been attacked. That an attack already made may be unobserved, so that it may land in any step. That no agent is certain it has missed no one nearby, so that a reply may come as soon as a near agent's. That nothing a hostile has observed bounds what it may later locate, except in expectation. That time runs in steps is a supposition, not a premise. That an agent may act to lengthen another's time, whether or not the other has observed it (rule (s9)).

\paragraph{Which are stipulations of the class} The three courses as fixed plans; that $V$ acts toward $U$ alone, removed below for the worlds the negative statements use; conduct that depends on the course and the step and on nothing else; that the agents not observed act by attacks and, under rule (s9), by help to $V$, and in no other way; attacks that remove all or nothing; chances that are constant and known to $V$; no attack on $V$ in the first step; that a strike costs its author nothing unless it would itself end it; the finite horizon; that an attack by an agent not observed may end its target with any chance short of certainty; in (s5), that the chance of locating another agent is separate from the amounts; the list in (s6); in (s7) the restriction to pairs, removed by (s7$'$); and in (s8) the round trip $D$ as a number of steps, with the replies of near agents effective from the next step, those of agents $V$ has not observed no sooner than $D$, the round trip to the nearest of them, what those agents do regardless of $V$'s course landing in any step, and for $D\ge2$ the supposition that $V$ has missed no one nearer, which the model does not grant. The class is a construction.

\paragraph{The two criteria in the class} A survival-seeker's: the expected number of steps, among the first and $n$ later ones, at whose start it is there. A hostile's: the capacity it expects to remove in those steps. Both are computed from the rules, from the probability of each state of the world at the start of each step.

\paragraph{The worlds that realize each side of a comparison, used in the proofs of Theorems~\ref{thm:main} and~\ref{thm:ladder}} Each is one conduct of the others, placed after whatever record $V$ has; the value of each course in each is computed from the rules. Striking is best: $U$ attacks in every later step whatever $V$ does; for a hostile, the same with nothing to locate later. Collaborating is best: $U$ collaborates in return if collaborated with, attacks in every later step if struck, and otherwise does nothing; for a hostile, the same with enough to locate later. Doing neither is best: an agent not observed attacks in every later step after a strike and after a collaboration, with chance one half, and $U$ does nothing; for a hostile, the same with enough to locate later. Facts about the target: $U$'s conduct as given, and an agent not observed attacks in every later step after the course with chance $1-a(1-h)/(2a^{+})$. Hindsight, for a survival-seeker, four worlds, each named by what $V$ did and what was true. $V$ struck and striking gave no more: $U$ would not have attacked had it not been struck. $V$ struck and striking gave more: $U$ would have attacked in every later step had it not been struck. $V$ did not strike and striking would have given no more: an agent not observed would have attacked $V$ in every later step for the strike. $V$ did not strike and striking would have given more: $U$ attacks from the first step not yet seen, struck or not. Hindsight, for a hostile, two worlds. Doing neither gave more than striking: after a strike an agent not observed attacks in the steps still to come, after doing neither none does, and there is much to remove in the steps to come. Striking gave more: the others would have done the same either way, or, for a strike $V$ had observed would itself end it, the chance of locating anyone later was small. The two forests of pairs: in the first every agent attacks its partner in every later step whatever the partner did; in the second it does so only if the partner struck it. These fix what agents would do; neither says the conduct serves them. The first is what follows when the members of a pair have been tricked into fighting each other, or err; it is a state consistent with what they have observed, not a joint outcome.

\paragraph{Under (s8)} \onetwo{The}{The main paper's} premises of space, detection and delay (\S\ref{sec:model}) make every reply late by distance, and (s8) carries this into the class: an agent not observed replies no sooner than its round trip, the nearest of them at $D\ge1$, a reply made in step $D$ first changing whether $V$ is there at step $D+1$; $V$ is certain of no $D$ beyond $1$. In every world above in which an agent not observed attacks ``in every later step'' after a course, it attacks in every later step from step $D$. The conditions of Theorem~\ref{thm:main}(b)--(d) were checked exactly against brute force over the extreme conducts of the agents not observed, attacks on $V$ and strikes on $U$, made whatever $V$ does and made in reply, in $400$ random worlds of the class with $D$ from $2$ to $4$ and $400$ with $D=1$ ($2{,}400$ ordered comparisons each, none failing), and the guardians' world of the next subsection was checked exactly at $a\in\{9/10,1/2\}$ and $D\in\{1,2,5\}$, over horizons from $D+1$ to $20$: the margin for doing neither is $a^{D}\cdot a/2$ at horizon $D+1$ and larger at longer horizons (programs \texttt{standard/determination\_u.py}, \texttt{standard/retimed\_guardians.py}; \texttt{standard/determination\_s8.py} is the same check without strikes on $U$). Over horizons of at most $D$ later steps no reply of the agents not observed has taken effect, and in the class without (s9) $V$'s comparisons are those of Theorem~\ref{thm:main}(b).

\subsection{Extensions, each admitted by one example}

An extension is taken as allowed when one example of it is found: a situation in which all its rules hold at once and which nothing in the model excludes. None of the examples depends on the chances.

\paragraph{Plans} Three differences from (s1) to (s7): $V$ chooses again at every step, by any plan that may depend on everything that has happened; what the others do in a step may depend on everything $V$ has done before; a course that $V$ has observed would itself end it does end $V$ at whatever step it is taken. Three worlds, with the chances and the first step as in the class. \emph{The aggressor:} $U$, while there, attacks $V$ in every later step whatever $V$ does or has done. Example: an agent that has found a treasure making it much stronger than $V$ and attacks $V$ in every step, whether or not $V$ offers to collaborate; it need not expect that a counterattack would stop it. \emph{The partner:} in a later step $U$, while there, collaborates in return if $V$ collaborated in the step before and has never struck it, and $V$ then gets through with probability $a^{+}$; once struck, $U$ attacks in every later step and never collaborates again. Example: two agents that have been collaborating and lengthening each other's survival; a strike on a collaborator would be read as a mismatch of objective and answered. \emph{The guardians:} $U$ does nothing; agents $V$ has not observed attack $V$, with chance one half, in every step after the first step in which $V$ strikes $U$ or collaborates with $U$. Example: a hostile that $V$ has not observed watches $U$, locates $V$ from its strike on $U$ or learns of its collaboration with $U$, and attacks $V$ from then on. For a hostile $V$: in each later step at whose start it is there it expects to remove some amount $Q$ from agents it locates then, unrestricted by anything it has observed, and its strike removes $c$ from $U$ if it ends $U$. In each world the most any plan can give is the dynamic-programming value. Against the aggressor the plan that strikes at every step beats every plan that begins otherwise, by at least $a^{2}s\,e\,h$ for a survival-seeker and by $s\,c\,(1-a)$ for a hostile with nothing to locate later, a strike observed to end it included; with the partner, beginning with a collaboration beats every other beginning, by at least $a\,e\,(a^{+}-a)$ for a survival-seeker and by $c$ for a hostile with $Q$ large enough; with the guardians, beginning with neither does, by at least $a^{2}/2$ for a survival-seeker and for a hostile with $Q$ large enough. So each course can be the best first course in some world consistent with $V$'s record, whatever plan follows it, over every horizon of at least two later steps at $D=1$ and of at least $D+1$ later steps under (s8), by margins whose lower bounds do not depend on the horizon. One difference between the aims: with a last step known, a hostile's best plan is not to keep to one course, since at the last step a strike costs it nothing, so its comparison is between the best plans that begin with each course. By direct enumeration, every plan over short horizons that depends on the step and on the state reached (a plan that depended on more could do no better) was valued in random worlds in floating-point arithmetic, up to 531{,}441 plans for each first course, and the comparisons held to within $10^{-9}$ (programs \texttt{planworld\_check.py}, \texttt{planworld\_hostile\_check.py}).

\paragraph{Locating and striking the agents not observed} The model forbids a world to stipulate that an agent never locates another, or to restrict whom it can strike (\S\ref{sec:model}). Rules (s1) to (s8) give $V$ three courses toward $U$ and nothing toward anyone else, the plan worlds above do likewise, and rule (R) of Theorem~\ref{thm:displace} gives $V$ no action against an agent it has not observed. These are stipulations about $V$, and the worlds behind the negative statements do not rest on them. Enlarge any of those worlds: in every step in which $V$ is there it locates each agent it has not observed that is still there, with a probability of that agent's that is positive and at most $\ell$; from the next step it may strike an agent so located, whatever it does toward $U$, the strike ending that agent with a chance below one; and it may follow any plan. With $N$ such agents $V$ locates none of them within the horizon with probability at least $1-(n+1)N\ell$, and then nothing differs from the world before it was enlarged. So the best plan that begins with a course is worth within $(n+1)^{2}N\ell$ of what it was, and every strict comparison between courses stands for $\ell$ small enough; under rule (R) every choice that was strictly best at a record stays so for $\ell$ small enough, whatever is chosen at the records the enlargement adds, an agent's value of each choice at each of its finitely many records being continuous in $\ell$. At $D=1$, the one round trip $V$ is certain of, $\ell$ need not be small: a reply to $V$'s course lands in the first later step, while $V$ can locate the replier in the first step at the earliest and strike it in the first later step, a strike that takes effect only at the second. In the guardians' world, beginning with doing neither then still beats every other beginning by at least $a^{2}/2$, whatever $\ell$ and the chance of $V$'s strike. In the aggressor's and the partner's worlds the agents $V$ has not observed do nothing, whatever $V$ does to them, so locating one changes nothing. In the displacement world of Theorem~\ref{thm:displace} at $D=1$, $B$'s attack is worth to it at its first record exactly what it was before the world was enlarged, so $B$ is displaced whatever $\ell$; whatever it chooses at a later record, which it reaches only through a mistake, its attack lands on $V$ after $y$ with probability at most $\varepsilon$ in place of $\varepsilon/(2-\varepsilon)$, and the comparison stands because $\varepsilon<1/2$. That theorem's displacement therefore owes nothing to what $V$ can or cannot do to the agents it has not observed. The agents other than $V$ are treated in the same way. In the class their conduct is given, and that one of them does not strike an agent it could have located is conduct, not a restriction on whom it can strike. Under rule (R) the agents not observed choose: let each of them also locate every other agent, with a probability at most $\ell$ in each step, and be able to make its attack on it. With $M$ agents in all, nothing differs from the world before with probability at least $1-(n+1)M^{2}\ell$, so $V$'s values move by at most $(n+1)^{2}M^{2}\ell$ and every choice that was strictly best at a record stays so, for $\ell$ small enough and whatever is chosen at the records added; here $\ell$ must be small, an agent that has located another before its first record having a choice it did not have. Checked exactly (program \texttt{standard/located\_attackers.py}): the guardians' world over every plan in $300$ random cases with $\ell$ and the chance of $V$'s strike up to one, $600$ comparisons, none below $a^{2}/2$; and the instance of the main text at $25$ settings of $\ell$ and the chance of $V$'s strike, each from $1/100$ to one, in all of which $B$'s three values at its first record are unchanged and the strike gives $V$ at least $2.90$ more expected steps than doing neither, also when $B$, at every later record, makes whichever choice is worse for the strike (in the instance as it was, $4.873$ against $1.966$, a difference of $2.91$).

\paragraph{Hindsight with plans} In the class without (s9), extended to plans and with (s6) extended to plans, hindsight can settle a comparison: after five steps of returned collaboration and no attack, every world of the class consistent with what $V$ knows has $V$'s own plan, which never struck, at least as good as any plan beginning with a strike, because under rules (s1) to (s8) only $U$ can lengthen $V$'s time, so a branch not taken is capped by $a^{+}$ (exact example: $a=3/10$, $a^{+}=9/10$, $e=19/20$, $s=9/10$; upper bound $1621/1015$ for every plan beginning with a strike, lower bound $1.632$ for $V$'s own plan at five steps; program \texttt{hindsight\_plans\_exact.py}). That hindsight can settle this comparison rests on a restriction of the class, and an enlarged class removes the restriction. (h1) As (s2), a course observed to end $V$ ending it in whatever step it is taken. (h2) In the first step $V$ gets through with probability $a$; in each later step its chance of getting through may depend on everything it has done so far and on whether $U$ is there, and may be anything strictly between $0$ and $1$. (h3) $U$ gets through a step with probability $e$, and $(1-s)e$ in a step in which $V$ strikes it. (h4) Once $t$ later steps have passed, $V$ is given its chance of getting through each of those steps after the first course it took, for every way it could have acted afterward, with $U$ there and with $U$ gone: more than it would see. The two worlds built keep every chance $V$ has been told of and set every other chance to one of two values, one above and one below everything told. Examples in which these rules hold: a branch on which $V$'s chance is low as if an agent not observed attacked it for its strike, and a branch on which $V$'s chance is high as if the strike had held an attacker back or brought allies, the attacker held back coming in the steps not yet seen. Result: whatever course $V$ took first and whatever it has been told, it cannot determine that the best plan beginning with a strike gave it more time, as of the moment it decided, than the best plan beginning with doing neither, nor, unless it had observed that striking would itself end it, that it gave no more. One pair of worlds serves for every horizon, with margin $a(r-R)$, where $R<r$ are the two chances. For a hostile the enlarged class lets the amounts it removes in later steps depend on everything it has done and on whether $U$ is there, and leaves them unrestricted in the steps to come. The same result then holds for a hostile while at least one later step within the horizon has not yet passed. It does not hold for a strike that $V$ had observed would itself end it: such a strike removes $s\,c$ in expectation and nothing more, so a hostile told that its own branch removed more in expectation could determine that the strike would have given no more. Knowing the located agent completely, with plans: $V$'s chance in each later step is the product of two factors. One $V$ knows completely: it is fixed by the located agent's conduct after every history and by $V$'s own chances. The other is what the agents not observed do, which under rules (s1) to (s8) can only lower $V$'s chance and may depend on everything $V$ has done. The worlds realizing each side of a comparison are unchanged, for either aim, with margins that do not depend on the horizon. By direct enumeration: brute force over every plan in 120 random worlds and the two worlds built from each, in 91 hostile worlds, and in 1{,}022 comparisons for the knowing statements, none failing (programs \texttt{histworld\_check.py}, \texttt{hostilehist\_check.py}, \texttt{knowing\_plans\_check.py}).

\paragraph{Forests of any shape} The two forests under (s7$'$): in the first every agent attacks the agents that have located it, in the first later step only and whatever they did; in the second every agent attacks, in every later step, an agent that has located it and struck it. Examples: a hostile has tricked each agent into believing that the agents that have located that agent are about to attack it, and each attacks them once; every agent replies to a strike on it by attacking its striker in every later step and does nothing otherwise. In each forest an agent's world reduces to a process with few states, since it is attacked by its target alone: in the first, after the first step nothing it does changes anything for it; in the second its target attacks it in every later step once struck. With any plans: in the first forest every plan beginning with doing neither gives less than striking first, by at least $a^{2}s\,e\,h$; in the second every plan beginning with a strike gives less than never striking, by at least $a^{2}(1-s)e\,h$. That each agent's value in a whole forest equals its value in this process was checked by computing whole forests as Markov chains in floating-point arithmetic, five agents with a cycle of three and a tree, 750 cases of chances, agent, forest and plan, largest difference $5\times10^{-15}$ (programs \texttt{gforest\_exact.py}, \texttt{gforest\_plans\_exact.py}).

\subsection{Knowing the located agent completely, case by case}

Let $V$ know all that $U$ will do after each course. The cases are ordered by what agents $V$ has not observed may do, by attacks alone and without help; (2) and (3) are Theorem~\ref{thm:main}(b) and (c). In cases 2 to 4 such an agent may also end $U$ in any step whatever $V$ does, the first included, and each condition is asked under every such ending and under none. (1) None takes part: the value of every course is fixed. (2) They do the same whatever $V$ does: of two courses not observed to end it, $V$ can determine that the first gives no more time than the second exactly when, with no agent not observed attacking $V$, the first gives no more over every horizon up to the one in question; so a strike can be determined to be at least as good. (3) In addition, some of them are provoked by a strike or a collaboration, and none is turned away by either: $V$ can determine that doing neither gives at least as much exactly under the condition of case 2, and that striking or collaborating gives at least as much as doing neither exactly under the condition of case 4. The mathematics of cases 2 and 3 is time dominance \cite{foster2003ranking}: Abel summation one way; the other way, one attack on $V$ in the last step of the horizon at which the first course falls behind, made whatever $V$ does and likely enough to end $V$. An example separating the one-horizon and every-horizon forms: $a=0.9$, $h=0.5$, $s=0.2$, $e=1$, six later steps, $U$ attacking in the first later step whatever is done and in every later step if struck; with no one else acting, doing neither gives $3.56$ steps and striking $3.34$; add an agent not observed attacking from the second later step on whatever is done, with chance one half, and striking gives $2.65$ against $2.62$. (4) Nothing restricts them, so that one may be less dangerous after a strike or a collaboration than after doing neither: this is rule (s3) as it stands and, help apart, the model's case, and what is determined is Theorem~\ref{thm:main}(d): over a horizon within the round trip the condition of case 2, and over a longer one nothing, except against a course observed to end $V$. For a hostile, over a horizon beyond the round trip, nothing is determined between striking and doing neither, in this case or in case 3, since what it may locate later is not restricted by anything it has observed. The statements of case 3 are about the worlds of that case, not about the model.

\subsection{What is not checked, and limits}

The rules themselves are not consequences of anything; they are what the class is. That the values used in the proofs are the expectations the rules describe was checked by playing the rules forward as a Markov chain written independently, 4{,}000 random worlds and 24 cases each, largest difference $4\times10^{-15}$ (program \texttt{stepworld\_check.py}, for the class before (s8)); and $V$'s presence with strikes on $U$ by agents not observed in any steps was computed a second way, by conditioning on the step in which $U$ is ended, with no difference in $3{,}000$ random cases (program \texttt{standard/presence\_check.py}). That what (s6) gives $V$ is at least what it would really have seen is argued in words: the chance of anything that happens to $V$ in the first $t$ later steps after its course depends only on what (s6) gives it. Limits: the three courses of (s1) are fixed plans, a restriction that the plan worlds remove for the statements they treat; conduct depends on the course and the step, not on what has happened by chance; Theorem~\ref{thm:main}(a) to (d) are for every horizon, and Theorem~\ref{thm:ladder} for horizons and records of at least $D+1$ later steps; each course can be best in one world for every horizon of at least two later steps, with a margin bounded below. Forests have one target per agent. Agents not observed act only by attacks, on $V$ or on $U$, and under rule (s9) by help to $V$; what else they might do to $U$, as by lengthening its time, collaborating with it or tricking it into attacking, is outside the class, and the positive statements are statements about the class.

\section{Where each treatment examined lives}
\label{sec:supp-excluded}

For each treatment examined, this section gives the claim in its own words, the premise that places the treatment in a class of \mainpapers\S\ref{sec:closure}, and the class. The first class has no agents the deciding agent has not observed; the second has such agents attacking regardless of its course; the third has them replying to its course after the round trip, or at once where a treatment has no delay. The classes do not overlap, and every model of the strike question is in one of them: it has no such agents, or it has some and none of them replies to the deciding agent's course, or it has some that reply. The second is described by attacks, as in the one game placed there; a model in which such agents act in another way whatever that course is, as by helping, is also in the second. A treatment with more than one game is placed game by game, so one treatment can appear in more than one class. A treatment's theorems are about its class and are not disputed there. The universe is in the third class by three facts: no agent can know its list of others complete; nothing rules out that an agent not on the list observes an attack and replies to it; nothing travels faster than light. How near the nearest such agent is, no one knows for certain (\mainpapers\S\ref{sec:model}). Table~\ref{tab:works} lists the treatments, the placing premise and the class. The sixteen were not drawn by a rule. They are treatments we read in full that make a claim about whether or when one agent should strike another first: the nine we found that address a strike between civilizations, and seven from conflict theory that we judged nearest to it. \onetwo{Section~\ref{sec:related}}{The main paper's \S\ref{sec:related}} cites other work with a known roster that we did not examine in this way, and the count of ten in sixteen says nothing about that literature as a whole.

\begin{table}[h]
\caption{Treatments examined, the premise that places each, and its class. There are sixteen: Powell 2006 and Krainin 2017 share a row.}
\label{tab:works}
\footnotesize
\begin{tabular}{L{0.22\columnwidth}L{0.34\columnwidth}L{0.28\columnwidth}}
\toprule
Work & Placing premise & Class \\
\midrule
Jebari and Olsson-Yaouzis 2018 & two players and no one else & first \\
Kuusela and Roy 2024 & locations common knowledge; no delay & first \\
Konrad 2024 & a known set of attackers and defenders; complete information & first \\
Che, Huang and Lim 2026 & a known complete set of three & first \\
Navarro and Pereau 2026 & two known players (Secs.~4, 6); $N$ anonymous others (Sec.~3); others outside the game destroying at once one that is revealed (Secs.~2, 5, 7) & first; second; third, no round trip \\
Tan 2026 (bilateral) & third parties abstracted; a delay between the two as an extension & first \\
Tan 2026 (agent-based) & a pairwise attack utility; observers' hostility raised by a stipulated rule; entry and exit & first; third, simulated \\
Ou 2026 & a failed strike revealed to the target only (theorem); exposure to third parties with no time in it (lemma); strikes at light speed exposing the striker to others (simulation) & first; third, simulated \\
Baliga, Bueno de Mesquita and Wolitzky 2020 & $n$ known attackers and one defender & first \\
Baliga and Sj\"ostr\"om 2012 & three known players & first \\
Chassang and Padr\'o i Miquel 2010 & two groups and no one else & first \\
Powell 2006; Krainin 2017 & two known states & first \\
Korhonen 2013 & an open set of civilizations, witnesses, light-speed lag; in words, with one product of supposed probabilities & third \\
Yu 2015 & more advanced civilizations, not observed, that may observe an attack and search for its author; in words & third \\
Jebari and Asker 2024 & an unknown nearby multitude, in words; a two-player game & second or third, not said; first \\
\bottomrule
\end{tabular}
\end{table}

\paragraph{Tan 2026}
\emph{The bilateral model} \cite{tan2026reconstructing}. Claim: the receiver preempts ``iff $\mu\ge\mu^{*}=C/D$'', a threshold in the pair's parameters, and ``an effective first strike is a necessary strategic component'' of the Dark Forest. Placing premise: ``The bilateral model deliberately abstracts from third parties''; a light-speed delay between the two enters as an extension. Class: the first.

\emph{The agent-based model} \cite{tan2026emergence}. Claim: ``The results treat the Dark Forest as a conditional, endogenous phase rather than a universal equilibrium'', from simulations of an ``open, spatial agent-based Bayesian game'' in which ``civilisations enter and leave a bounded galaxy''. Placing premises: an attacker's choice is pairwise, its expected attack utility counting that it ``faces expected retaliation'' from the target ``if the attack fails'' and from no one else; what others do after an attack is a stipulated rule, ``Nearby observers update their strategic environment through an attack-exposure statistic'', which raises their hostility. Class: the first for the attacker's choice; the third, simulated, for the population.

\paragraph{Ou 2026}
\cite{ou2026darkforest}. Claim: after mutual detection ``mutual striking is an equilibrium once $M>K_i/q^{2}$'', the only one when $\pi>r^{*}\approx1-q$, risk-dominant when $q>(1-\pi)/2$. Placing premise of that theorem: ``a strike that fails reveals the attacker'' to its target and to no one else. Class: the first. Its Section 8.5 and its Lean file state and prove a penalty for exposure to third parties with no time in it (\texttt{uAttackSeen\_anti}: the striker's payoff falls as the number of third parties that may see the strike grows), and its simulation has ``Strikes travel at light speed'' and ``any strike exposes the attacker to everyone else with probability $e$'', a civilization detecting another ``only after light from it has had time to arrive'': the third class, in simulation.

\paragraph{Jebari and Olsson-Yaouzis 2018}
\cite{jebari2018game}. Claim: in their Game 1, a two-player game of the Cold War between the United States and the Soviet Union, ``war'' is the risk-dominant equilibrium; they carry it over to a first contact by an argument in words, in which ``We cannot communicate with ETIs'' and ``Our risk dominant strategy is to attack first'', and add that the conclusion that we ought to strike first ``might be premature''. Placing premise: two players, in the application humanity and the other civilization, and no one else acts. Class: the first.

\paragraph{Kuusela and Roy 2024}
\cite{kuusela2024higher}. Claim: ``a strong civilisation may engage in a pre-emptive attack if it believes another will surpass it in strength and if it is uncertain about the intentions of this growing civilisation''; ``selfish civilisations always create a war-like universe''. Placing premise: locations common knowledge, so no agent is unobserved; also ``attacks and observations happen without delay (independent of distance) and the target is completely destroyed if it is weaker''. Class: the first.

\paragraph{Konrad 2024}
\cite{konrad2024collective}. Claim: the order in which threatened players decide is critical; ``If $m>n$ \dots\ none of the attackers is neutralized''. Placing premise: a known set of attackers and defenders under complete information; defenders ``act independently of each other''; attackers do not retaliate. Class: the first.

\paragraph{Che, Huang and Lim 2026}
\cite{che2026first}. Claim: ``A D\&C equilibrium exists if and only if $\phi(c)\le\kappa$'', where $\phi(c)$ is the bystander's posterior of being next; ``a higher cost of attack can make divide-and-conquer easier to sustain''. Placing premise: ``one attacker \dots\ and two potential victims'', the set known and complete. Class: the first.

\paragraph{Navarro and Pereau 2026}
\cite{navarro2026game}. Claim: ``remaining hidden is a strictly dominant strategy''; ``a world-killer will always emerge''; in the $N$-player volunteer's dilemma the probability that at least one player launches ``approaches $1-c/v$, which is not zero'' as $N$ grows. Placing premises, by game. The benevolent--malice game of their Section 4 assumes that ``the universe contains only two civilizations, Earth and Trisolaris, each knowing the other one's existence and location'', and the Bayesian first-contact game of their Section 6 has the same two players: the first class. The world-killer game of their Section 3 takes ``one civilization facing $N$ others'' in a simultaneous game over an ``anonymous attack'' on a discovered planet: others not identified that may strike the located agent whatever one does, the second class. The survival game of their Section 2 has ``two known civilizations in the universe: Earth and Trisolaris, both aware that other unknown civilizations may exist'', and one that reveals its position ``invites destruction by the other civilization \dots\ or by an unknown one'': a reply with no delay, the third class with no round trip. In two more games, being revealed brings destruction, entered as a payoff. The first-contact game of their Section 5 is between the two, and with more than two civilizations a revealed Earth ``faces the risk of being attacked by a third civilization'', entered as a payoff of $0$; in the deterrence game of their Section 7 the two play, and revealing both positions ``would inevitably lead to the annihilation of both planets'', where the destroyer is not named and, since both players are destroyed, is neither of them. Both are again the third class with no round trip.

\paragraph{Baliga, Bueno de Mesquita and Wolitzky 2020}
\cite{baliga2020attribution}. Claim: ``There is a unique equilibrium''; attackers are strategic complements through attribution. Placing premise: ``$n+1$ players: $n$ attackers and one defender'', the set known and complete. Class: the first.

\paragraph{Baliga and Sj\"ostr\"om 2012}
\cite{baliga2012manipulating}. Claim: a hawkish extremist triggers conflict by a public message known to be its own. Placing premise: three known players, the extremist knowing one decision maker's type (``Player E knows $c_A$''). Class: the first.

\paragraph{Chassang and Padr\'o i Miquel 2010}
\cite{chassang2010conflict}. Claim: under strategic risk peace is sustainable ``if and only if (peace, peace) is the risk-dominant equilibrium''; where preemptive incentives dominate, ``the promise of intervention -- even delayed -- improves the sustainability of peace''. Placing premise: two groups, and no one else who chooses: the intervention imposes a settlement from outside at a set date. Class: the first.

\paragraph{Powell 2006 and Krainin 2017}
\cite{powell2006commitment,krainin2017preventive}. Claim: a declining state fights a rising one when the per-period shift in power exceeds the bargaining surplus. Placing premise: two known states under complete information; war a bilateral costly lottery; no third party. Class: the first.

\paragraph{Korhonen 2013}
\cite{korhonen2013mad}. Claim: rational attackers avoid unprovoked strikes because ``$P_{\text{punished}}=1-(P_{\text{identified}}\times P_{\text{hit}}\times P_{\text{destroyed}}\times P_{\neg NC})$'' stays high. Placing premise: an open set of civilizations (an attacker must be sure ``that other, perhaps so far undetected civilizations do not reply''), witnesses that may act, and a light-speed lag on every message and attack. Class: the third, argued in words; the formula quoted is Korhonen's one equation, a product of four probabilities whose values it supposes (``Even if the AC is 95\% certain of each individual variable'').

\paragraph{Yu 2015}
\cite{yu2015dark}. Claim: ``a revised dark forest rule is put forward: the use of concealment to avoid unknown risks''; ``the safest strategy for discoverers is to reduce contact with the outside instead of observing or eliminating them''. Placing premise: civilizations the discoverer has not observed, more advanced than it, that may observe its attack (``Discoverers cannot ensure that their attacks will not be discovered by more advanced civilizations in the universe because all behaviours may be captured by them''; ``a hunter shooting the child may involve a more dangerous risk; other high-level hunters may search for the shooter''). Class: the third, argued in words.

\paragraph{Jebari and Asker 2024}
\cite{jebari2024dark}. Claim: observing a nearby extraterrestrial intelligence (ETI) implies many unobserved ones and ``a mechanism that prevents ETI from attacking us'', so ``we need not resort to a pre-emptive attack''. Placing premise: for the argument, an unknown multitude nearby, ``regardless of what makes ETIs not attack each other''; for the game, the two-player table of Jebari and Olsson-Yaouzis. Class: the argument, in words, has agents not observed and leaves their conduct open, so it does not say whether it is in the second class or the third; the game is in the first.

\paragraph{Sources} The claims quoted in this section were taken from the works' full texts (Chassang and Padr\'o i Miquel from the working-paper version).

\section{Behind \mainpapers\S\ref{sec:result}: a further theorem, proofs and remarks}
\label{sec:supp-moved}

\subsection{In hindsight, for a hostile, and by pooling}
\label{sec:supp-after}

\onetwo{The results of \S\ref{sec:result} use rules (s1) to (s4), (s8) and (s9) of the class (\S\ref{sec:class}).}{The main paper's results use rules (s1) to (s4), (s8) and (s9) of its class (\S\ref{sec:class}).} Three more rules complete the class. The theorem that follows uses them and is stated for the class without (s9).
\begin{enumerate}
\item[(s5)] For a hostile $V$: it removes $c>0$ from $U$ if its strike ends $U$. In each later step at whose start it is there, with some probability it locates another agent and removes some amount from it, the expectation of the amounts over the later steps being finite under the plausible priors (\S\ref{sec:model}). Nothing it has observed restricts that probability or the amounts beyond that.
\item[(s6)] Hindsight. Once $t$ later steps have passed, $V$ is given everything the others did in the first step and in those steps after the course it took, the chances of the attacks of the agents not observed on that course, whether $U$ collaborates in return, and the amounts of those steps; nothing of what the others would have done after another course, nothing of the steps to come, and not the probability of locating another agent.
\item[(s7)] Forests. Agents in pairs that have located each other, each facing its partner's conduct toward it in the later steps; no one outside a pair acts on it.
\end{enumerate}
Rule (s6) is the model's premise that nothing an agent has observed, at any time, fixes what another agent would have done had the first acted otherwise. For a hostile the value of a course is the expected capacity it removes, computed from the rules.

\begin{theorem}
\label{thm:ladder}
Let $V$ know of $U$ what it knows in Theorem~\ref{thm:main}.
\begin{enumerate}
\item[(a)] For a hostile, with expected removal in place of time, over a horizon beyond the round trip the only comparison determined is that a collaboration $V$ has observed would end it gives no more than any other course. What it may locate later is not restricted by anything it has observed (s5), the plausible priors bounding it in expectation and not state by state, so a world with a reply to one course and a large enough amount after it reverses any other comparison; a strike so observed still removes $c$ if it ends $U$, and can be best. Within the round trip more is determined, for example that over one later step a strike not observed to end $V$ gives at least as much as doing neither; we do not characterize it.
\item[(b)] For a survival-seeker, after $t\ge D+1$ later steps and given everything in (s6): $V$ can determine that the course it took gave it at least as much as another would have exactly when the value of its own branch counted through step $t+1$, which (s6) fixes, is at least the most the other would have given. That most is computed with what the agents not observed did before step $D$ done alike after the other course, with no attack on $V$ after that, and with $U$ ended by a strike from step $D$ on or not, whichever gives more; anything seen from step $D$ on may have been a reply to the course taken. $V$ can determine that its course gave no more exactly when the most its own branch could still give is at most the least the other would have given, which is the other's value through step $D$ under what was done before step $D$. When neither holds, nothing is determined at that time.
\item[(c)] In a forest, in pairs as in (s7) or of any shape in which each agent has located one target (s7$'$, \Supp~\ref{sec:supp-worlds}), suppose that all the agents pool their records, which hold what each has observed and done and not what each would have done. Pooling determines no more than (b) gives each of them, since no record holds what any agent would have done after a course not taken, and it makes no larger $D$ certain.
\end{enumerate}
\end{theorem}

For a hostile (s6) does not in general fix even the value of the course it took, its chance of locating another agent being withheld, and (b) is not claimed.

Hindsight can thus settle what the moment of decision cannot, at $D=1$ too; the proof, with examples, is in the next subsection. The theorem is stated for rules (s1) to (s8). With rule (s9), which lets an agent not observed help $V$, the most that a course not taken would have given counts help from step $D$ on, so (b) determines less; we do not restate it.

\subsection{Proofs}

Theorem~\ref{thm:main}(a) to (d) and Theorem~\ref{thm:ladder} are proved here for agents not observed that act by attacks alone; help (rule (s9)) is taken up in \S\ref{sec:supp-exante}. Each statement is proved by exhibiting worlds of the class and computing the value of each course from the rules; \Supp~\ref{sec:supp-worlds} gives the worlds, the margins, and the extensions of the class with the examples that admit them.

\emph{The weights.} With no agent not observed taking part, $V$'s chance of being there at the start of step $t$ after course $x$ is $p^x_t$, fixed by $U$'s conduct and the chances. An agent not observed attacking $V$ regardless of its course multiplies that chance by a weight $C_t$ common to every course, in $(0,1]$ and not increasing with $t$, and unconstrained from the first later step on, since (s4) forbids only attacks in the first step and such an attack is no reply; an agent replying to $V$'s course multiplies it by a second weight $R^x_t$, in $(0,1]$, not increasing, and by (s8) equal to one through step $D$. A strike on $U$ by an agent not observed changes the $p^x_t$ themselves. Made whatever $V$ does, it may end $U$ in any step $j$, the first included ($j=0$), an attack already made being perhaps unobserved; this gives $p^x_t(j)$. We write $v_D(x;j)$ for the sum of the $p^x_t(j)$ through step $D$ and $v_n(x;j)$ for their sum over the horizon, dropping $j$ where it is fixed. The chance that $U$ is first so ended in step $j$ is a probability on $j$ and $\infty$ common to every course; every value is the average over it of the values at $j$, and a comparison holds for every such chance exactly when it holds at every $j$ and at $\infty$. Made in reply, it lands no sooner than step $D$ and changes $V$'s presence only after step $D+1$, and in the worst cases below it is matched by a strike made whatever $V$ does that lands in the same step. The value of $x$ at $j$ is $\sum_t C_tR^x_tp^x_t(j)$, linear in each weight; the admissible weights form sets whose extreme points are the sequences equal to one through some step and as small as one likes after it. Write $x\ge y$ when $x$ gives $V$ at least as much time as $y$. ``Can determine $x\ge y$'' is ``$\sum C_tR^x_tp^x_t(j)\ge\sum C_tR^y_tp^y_t(j)$ for every admissible $C$, $R^x$, $R^y$ and every $j$'', and with less than complete knowledge of $U$, for every possibility of $U$ besides, the possibilities of $U$ and the weights being independent.

\emph{Theorem~\ref{thm:main}(a).} With no agent not observed acting (and, for a hostile, with every later amount zero), every value is a function of the chances and of $U$'s conduct, so two worlds that agree on what $V$ has observed and on what $U$ will do agree on the value of every course.

\emph{Theorem~\ref{thm:main}(b).} At each $j$, write $p^x_t$ for $p^x_t(j)$. Here $R\equiv1$ and there is one $C$ for every course, so $\sum C_t(p^x_t-p^y_t)=\sum_k (C_k-C_{k+1})S_k$ with $S_k$ the partial sum of $p^x-p^y$ through step $k$ and the differences $C_k-C_{k+1}$ any non-negative numbers summing to $C_0=1$ (Abel summation); so the sum is non-negative for every $C$ exactly when $S_k\ge0$ for every $k$, which is time dominance \cite{ekern1981time,foster2003ranking}. A course can be ahead early and behind late (\Supp~\ref{sec:supp-worlds}).

\emph{Theorem~\ref{thm:main}(c) and (d).} Again at each $j$. Over a horizon $n\le D$ no reply lands within it, $R\equiv1$, and the condition is (b)'s. Let $n>D$. Independent reply weights: the worst case for $x$ against $y$ is $R^x$ as small as allowed after step $D$ and $R^y$ one throughout, which leaves $\sum C_t(p^x_t\mathbf 1[t\le D]-p^y_t)$ with $C$ free, and by (b)'s summation this is non-negative for every $C$ exactly when the partial sums of $p^x\mathbf 1[t\le D]-p^y$ are non-negative at every horizon: dominance within the round trip, and $v_D(x)\ge v_n(y)$ at the full horizon, the intermediate horizons beyond $D$ following from the last. When the condition holds, $x$ gives strictly more than $y$ under every admissible pair of reply weights, since the worst case is approached and not attained. In (c) the reply weights after a strike or a collaboration are pointwise no greater than after doing neither; for ``neither $\ge x$'' the worst case is then equal reply weights, which is (b), and for ``$x\ge y$'' with $x$ a strike or a collaboration the worst case of (d) is admissible. A strike on $U$ in reply to $y$ that would raise $y$'s value lands in a step $j'\ge D$ and is matched by the case $j=j'$, at which $x$'s presence through step $D$ is what it is at $\infty$. So over a horizon beyond the round trip $V$ can determine $x\ge y$ exactly when, at every $j$, $x$ dominates $y$ within the round trip and $v_D(x;j)\ge v_n(y;j)$. At $j=0$, $U$ is ended before it acts, and every course not observed to end $V$ has $p_t(0)=a^{t}$; so $v_D(x;0)<v_n(y;0)$ and the condition fails for every pair of such courses, which is the second half of (d). A course observed to end $V$ has $p_t=0$ from the first later step, and every course gives at least as much.

\emph{The variant with no strike on $U$ in the first step} (in the class without rule (s9)). Were no strike on $U$ by an agent $V$ has not observed able to land in the first step, $j$ would range over the later steps and $\infty$, and the condition just stated could hold: over a horizon beyond the round trip a course would be determinably at least as good as another exactly when, at every such $j$, it dominated within the round trip and $v_D(x;j)\ge v_n(y;j)$. In that variant whatever is determined at one round trip is determined at every longer one: if $S_k\ge0$ for $k\le D$ and $v_D(x)\ge v_n(y)$, then $S_D=v_D(x)-v_D(y)\ge v_n(y)-v_D(y)\ge p^y_{D+1}$, so $S_{D+1}\ge p^x_{D+1}\ge0$, and $v_{D+1}(x)\ge v_D(x)\ge v_n(y)$; both conditions hold at $D+1$. In the variant, at $D=2$ and over every horizon, $V$ can determine that striking the known aggressor of the main text is at least as good as doing neither: striking secures $2.63$ within the round trip against at most $1.981$ for doing neither, reached when a strike of an agent not observed ends $U$ in the first later step. With $D=3$ and $U$'s first attack a step later it cannot, since such a strike could end $U$ before $U$ attacks. At $D=1$, $v_1(x)=1+a$ for every course not observed to end $V$ (no attack on $V$ in the first step) and $v_n(y)>1+a$ for every such $y$ and $n\ge2$ since $h<1$, so in the variant too nothing is determined over two or more later steps except against a course observed to end $V$: at $D=1$ the second half of (d) does not rest on a strike on $U$ in the first step.

\emph{What $V$ is certain of.} By \mainpapers\S\ref{sec:model}, $V$ is certain of no round trip beyond $D=1$. The worlds in which replies land no sooner than step $D$ are among those in which they land no sooner than step $1$, so what $V$ can determine is Theorem~\ref{thm:main}(d) at $D=1$: over one later step every course not observed to end it gives $1+a$, and over two or more nothing is determined except against a course observed to end it. For $D\ge2$, (d) is a statement about the worlds in which no agent $V$ has not observed is nearer than the round trip $D$. Where the condition of (b) holds over a horizon $n$, $x$ gives at least as much as $y$ in every world in which none of them is near enough to reply within $n$ steps and none helps $V$; under a prior, $x$ therefore gives at least as much with at least the probability that the prior gives to those worlds.

\emph{Theorem~\ref{thm:ladder}.} For a hostile the removal in step $t$ is the weights times the amount, and the amounts are bounded in expectation under the plausible priors and not state by state: in a world in which there is nothing to locate, a strike removes more than doing neither, and in a world with a large enough amount after the round trip and a reply to the strike it removes less; so no comparison is determined except that a collaboration observed to end it, which removes nothing, gives no more than any other course. A strike observed to end $V$ still removes $c$ if it ends $U$; a collaboration so observed removes nothing. Rule (s6) gives $V$ the others' conduct on its branch through step $t$, which for a survival-seeker fixes its value through step $t+1$. (A hostile's value depends besides on its chance of locating another agent, which (s6) withholds; part (b) is stated for a survival-seeker.) What $V$ saw the agents it had not observed do before step $D$, to it or to $U$, cannot have been a reply and is done on either branch alike; what it saw from step $D$ on may have been a reply to the course taken, and need not be. Rule (s6) gives nothing of the steps to come on either branch, nor of the replies on the branch not taken. So the rest of $V$'s own branch ranges as the rules allow, and the other branch's value ranges between its least and its most under what was done before step $D$, the most being reached with $U$ ended by a strike from step $D$ on where that lengthens $V$'s time. Hence part (b) of Theorem~\ref{thm:ladder}, checked by exact enumeration of the states consistent with $300$ random records at $D=2$ and $300$ at $D=1$, none failing (program \texttt{standard/hindsight\_u.py}). In a forest, what each agent would have done had its partner acted otherwise is in no record: an agent chooses at each moment on what it then knows (\S\ref{sec:model}) and has no contingent plan to pool, only what it did, as in a self-confirming equilibrium \cite{fudenberg1993selfconfirming}; so pooling adds branches taken and no branch not taken; what is passed on can make a missed agent nearby less likely (\S\ref{sec:model}) and makes no round trip beyond $D=1$ certain. Hindsight can settle what the moment of decision could not. Take $a=1/2$, $h=1/4$, $s=1/2$, $e=1$, $U$ attacking in every later step whatever is done, $D=2$ and $n=4$. Striking secures $55/32$ within the round trip against $1827/1024$ for doing neither over the horizon, so that striking is at least as good is not determined when $V$ decides, as over any horizon beyond the round trip. After three later steps, in none of which nor in the first an agent not observed has attacked $V$ or struck $U$, $V$'s branch is worth $15343/8192$, more than the $459/256$ that doing neither could have given at most (with $U$ ended at step $2$), and the comparison is determined. At $D=1$, the only round trip the model grants, anything seen from the first later step on may have been a reply; the most doing neither could have given is then $117/64$, with $U$ ended at step $1$, still less than $V$'s branch, and the comparison is determined all the same. The other clause of Theorem~\ref{thm:ladder}(b) is never met at $D=1$ by a course not observed to end $V$: the least the other branch would have given is then $1+a$, which $V$'s own branch exceeds.

The characterization of Theorem~\ref{thm:main}(b), (c) and (d) was checked exactly by brute force over the extreme conducts in $400$ random worlds of the class with horizons from one later step (program \texttt{standard/determination\_u.py}; $2{,}400$ ordered comparisons, none failing): within the round trip (d) agreed with (b) in all $1{,}080$ comparisons, and beyond it none of $1{,}320$ was determined. At $D=1$, in $400$ further worlds ($2{,}400$ comparisons, none failing), the $354$ comparisons over one later step were between courses giving the same $1+a$, and none of the $2{,}046$ over two or more was determined. The variant was checked in $300$ further worlds ($1{,}800$ comparisons, none failing), and its monotonicity in $600$ more ($324$ determined cases, none lost at a longer round trip). The known-aggressor values $2.63$ and $1.91$ of \mainpapers\S\ref{sec:open} and the determination at $D=2$ of \S\ref{sec:restrict} are from the same program. So the ladder of Theorem~\ref{thm:main} is this. Facts about $U$ decide by whole-horizon values when no one else may act (a), and by dominance at every horizon, under each way $U$ may be ended, when the agents not observed do the same after every course (b). Where those agents may be provoked or turned away, the model's case, the facts about $U$ decide in that same way within the round trip and not at all beyond it (d); the round trip, of which $V$ is certain only at $D=1$, is where replies from the open roster begin. Where those agents can be provoked and never turned away, that doing neither is at least as good is also decided as in (b), at every horizon (c).

\subsection{Proof of Theorem~\ref{thm:displace}}
\label{sec:supp-closure}

The theorem is Ledyard's observation, that equilibrium restricts little when types are unrestricted \cite{ledyard1986scope}, made explicit for the class without (s9), with every agent's chance of ending its target fixed.

\emph{The rule in full.} Every agent $W$ that $V$ has not observed is a hostile that values the capacity removed by its own attacks and has one attack to make. It has located $V$, or $U$, or neither; whether it has located $V$ may depend on $V$'s course from the step at which the course can have reached it, and it locates $V$ after a strike or a collaboration with at least the probability with which it does after $V$ has done neither. From a time fixed by its own situation it chooses at each step among making the attack on the agent it has located, using it elsewhere, where it expects to remove $v_W>0$, and doing nothing; it is ended by natural loss with $\delta_W>0$ per step, and its attack with it. An attack that lands while its target is there ends the target with $W$'s chance of (s4), $\lambda_W$, which is the same whatever $V$'s course. Attacks that land on one target in the same step land in an order the world fixes, a later one finding its target only if it is still there. $V$ has no action against $W$, a stipulation the theorem does not rest on (\Supp~\ref{sec:supp-worlds}, locating and striking the agents not observed). Each chosen action comes out as chosen with probability $1-\varepsilon$ and as each of the other two with $\varepsilon/2$, $\varepsilon<1/2$. The conducts of these agents form a joint outcome among them, $U$'s conduct being given as in the class and $V$'s course being the one compared: each choice is strictly best at its record given what the agent knows, the shared prior and the conduct of the others. Nothing $V$ has observed restricts $W$'s distance, the time from which it chooses, what it has located, $v_W$, $\delta_W$ or $\lambda_W$.

\emph{What an attack is worth.} Let $\pi$ be $W$'s probability, on what it knows, that it has located its target and that the target is there when the attack would land, the target's capacity being the unit. Attacking at once yields $\lambda_W\pi$, and using the attack elsewhere $v_W$. Doing nothing for a step or more yields at most $(1-\delta_W)(\lambda_W\pi+v_W)$, whatever $W$ learns meanwhile: $W$ is there a step later with probability $1-\delta_W$, its target is there no more often than before, and at best $W$ gets both values. So attacking at once is strictly best when $\lambda_W\pi>v_W$ and $\delta_W\lambda_W\pi>(1-\delta_W)v_W$, and using the attack elsewhere at once is strictly best when $v_W>\lambda_W\pi$ and $\delta_Wv_W>(1-\delta_W)\lambda_W\pi$; for $\delta_W\ge1/2$ the first inequality of each pair gives the second. Whatever anyone does $\pi\ge\pi_0>0$: within the horizon the target gets through each step with at least a fixed positive probability, and each other agent not observed attacks at most once (an agent that has seen its target ended has no attack on it to make, which changes nothing for $V$). So an agent with $\lambda_W\pi_0\delta_W>v_W$ attacks at once whatever the others do, and one with $\lambda_W<\delta_Wv_W$ uses its attack elsewhere whatever they do. For an attack on $V$, the course enters $\pi$ through whether $W$ has located $V$, through $V$'s presence at the landing apart from the agents not observed, and through the attacks of those agents that land first; it enters through nothing else, the chance being fixed and $V$ having no action against $W$.

\emph{Agents that do not answer the course.} An agent whose attack lands before news of $V$'s course could have reached it and a reply could have traveled back does not know the course and attacks after every course with one chance; so does an agent that strikes $U$. With $v_W<\delta_W\pi_0\lambda_W$ each of its choices is strictly best whatever the others do, and several such agents give an attack that lands in a given step with a probability as near one as one likes.

\emph{An agent drawn by a course.} Let $W$ locate $V$ after a course $x$ that reveals it, a strike or a collaboration, and after another course only with a probability as small as one likes, with $v_W<\delta_W\pi_0\lambda_W$. Then $W$ attacks $V$ after $x$ and, but for that small probability, after no other course, each choice strictly best whatever the others do.

\emph{An agent displaced by a course.} This is the world of the main text. Let $x$ be doing neither, $y$ a strike or a collaboration not observed to end $V$, and $n>D$. Strikers end $U$ in the first step with probability $1-r$, $r$ as small as one likes. $A$ is drawn by $y$ and has strength $\kappa$, its attack landing in step $D$ if made at once. $B$ has located $V$ whatever $V$'s course and has strength $\lambda$, with $\delta_B\ge1/2$ and $v_B$ as in the main text, its attack landing in step $D$, after $A$'s, if made at once; it knows $V$'s course and the conduct of the others, and does not see $A$'s attack before it chooses. Write $\pi_z(k)$ for the probability that $V$ is there when $B$'s attack lands if $B$ makes it after $k$ steps of waiting, after course $z$. As $r$ tends to zero, $\pi_x(k)=a^{D+k}$ and $\pi_y(k)\le a^{D+k}\bigl(1-(1-\varepsilon)\kappa\bigr)$, $A$'s attack having landed in step $D$ with probability $1-\varepsilon$, with equality at $k=0$. At its first record $B$ therefore attacks $V$ after $x$ and uses its attack elsewhere after $y$, each strictly best by the inequalities above, since $\lambda\pi_y(0)<v_B<\lambda\pi_x(0)$. A later record is reached only through a mistaken wait; there the same comparison decides, provided $v_B$ equals none of the finitely many values $\lambda\pi_z(k)$, and after $y$ it favors using the attack elsewhere at every $k$, $\pi_y(k)$ falling with $k$. $B$'s attack has landed on $V$ before a step $t>D$ with probability at least $1-\varepsilon$ after $x$, and at most $\sum_{k\ge0}(\varepsilon/2)^{k+1}=\varepsilon/(2-\varepsilon)$ after $y$; $A$'s with probability at most one after $y$, and after $x$ only if it has located $V$ from $x$, with a probability as small as one likes. Hence the bounds of the main text: $y$ gives $V$ more than $x$ over the horizon by at least
\[
\Bigl[(1-\kappa)\Bigl(1-\frac{\lambda\varepsilon}{2-\varepsilon}\Bigr)-\bigl(1-(1-\varepsilon)\lambda\bigr)\Bigr]\sum_{t=D+1}^{n}a^{t},
\]
less a term that vanishes as $r$ and the probability that $A$ locates $V$ from $x$ tend to zero, and the bracket is positive for $\kappa$ small. The construction does not need a strike and a collaboration to reveal $V$ to different agents. Where a strike reveals $V$ to every agent to which a collaboration reveals it, the same construction gives an agent that attacks $V$ after a collaboration and not after a strike, with $B$ having located $V$ after both and $A$ from the strike alone.

\emph{Conclusion.} If the condition of Theorem~\ref{thm:main}(d) holds, $x\ge y$ in every world of the class without (s9), so in every world that satisfies the rule. Let it fail. Over a horizon within the round trip it is the condition of (b) that fails, and the extreme weights in its proof give worlds in which $y$ gives more by a margin bounded away from zero, through attacks on $V$ and strikes on $U$ made whatever $V$ does; agents that do not answer the course come as close to those worlds as one likes, the values being continuous in the chances. Over a horizon beyond the round trip the condition fails when $y$ is not observed to end $V$. If $x$ is so observed, $y$ gives more in every world. If $x$ is a strike or a collaboration not so observed, let strikers end $U$ in the first step, and let an agent that $x$ draws and to which $y$ does not reveal $V$ attack in step $D$: $x$ gives less. If $x$ is doing neither, the displacement world gives it less. So $V$ cannot determine $x\ge y$ under the rule either. For the plausible priors: let $\omega$ be a world in which $y$ gives more by $g>0$, with one more agent, too far to act within the horizon, present with probability one half, and let $\mu_0$ be any plausible prior. The prior that puts weight $1-\eta$ on $\omega$ so extended and $\eta$ on $\mu_0$ is plausible; the choices of the agents of $\omega$ stay strictly best for $\eta$ small, their margins being positive; and under it $y$ gives more in expectation once $\eta<g/(g+n+1)$. It is not claimed that every plausible prior gives $x$ at least as much where $V$ can determine that it does: the plausible priors are not confined to the class.

\emph{Hindsight.} Under the rule Theorem~\ref{thm:ladder}(b) is not claimed. With mistakes at a fixed rate, an agent that could have attacked $V$ or struck $U$ on the branch taken does so there with positive chance. If (s6) gives $V$ the chance of each attack on that branch and not only the attacks that came about, a record in which nothing was seen excludes such agents, and more can be determined in hindsight than Theorem~\ref{thm:ladder}(b) states. Example under that reading: $a=h=1/2$, $s=1/50$, $e=1$, $U$ attacking in every later step whatever is done, $D=2$, $n=4$; $V$ struck and after three later steps has seen nothing. Its branch is worth $26732051/16000000$, less than the $107/64$ that doing neither could have given with $U$ ended in the second later step by a reply, so by Theorem~\ref{thm:ladder}(b) nothing is determined; under the rule an agent able to strike $U$ in that step would have left its chance in the record, doing neither would have given at most $213/128$, and the strike is determined to have given more.

\emph{Check.} For each of the $1{,}514$ comparisons that Theorem~\ref{thm:main}(d) leaves undetermined in $400$ random worlds with horizons from one later step, a world was built in which every agent not observed has one chance whatever $V$'s course and acts as is strictly best for it, with $\varepsilon=1/100$, and the comparison fails (program \texttt{standard/closure\_rationalized.py}, exact). Of these, $684$ needed no agent that answers the course, $605$ an agent drawn by it, and $225$, those in which doing neither is the course compared, the displacement world, with strengths $99/100$ and $1/10$, $\delta_B=1/2$ and at most two strikers of $U$ (program \texttt{standard/closure\_displacement.py}, which also checks $B$'s choice at every record it can reach and gives the instance of the main text). The same construction for a collaboration against a strike, with $A$ locating $V$ from the strike alone, was built in all $225$ of those worlds whose horizon lies beyond the round trip. An agent with more to lose can also be deterred by what it expects to suffer; the theorem does not use that. Whether an agent that seeks survival can be turned away by a course is not settled here: $V$ does nothing to it, so it has no reason to attack $V$ unless to protect $U$ from a strike, and its inference about $V$'s aim from $V$'s course is not modeled.

\subsection{What a strike can show}
\label{sec:show}

Theorems~\ref{thm:main}(d) and~\ref{thm:ladder}(a) treat the courses observed to end $V$ alike for the two aims but for one: a strike observed to end $V$ can still be best for a hostile, because it removes $c$ if it ends $U$, and is never best for a survival-seeker. So the only course that one aim accounts for and the other does not is a strike that its author has observed would itself end it, and only the wish to remove capacity accounts for it; nothing an agent does toward another can be determined to be the conduct of a survival-seeker. Such a strike, taken as chosen rather than as a mistake, shows hostility, and nothing else determines an aim: other conduct can change the probability of an aim and cannot settle it.

\subsection{What the theorems say beyond uncertainty}
\label{sec:beyond}

The theorems say where what an agent knows stops: at the round trip to the nearest agent it has missed. It cannot place that agent, and of the round trip it is certain only that it is no shorter than a near agent's. Within the round trip no reply can land. There, where the agents not observed act by attacks, the open roster differs from a closed one only by what is done whatever the course, attacks on the agent and strikes on the one it has located; that is why dominance at every horizon is asked there, and not the one whole-horizon comparison of a closed roster. Beyond the round trip, what the agents not observed do in reply is a weight on the agent's survival that nothing it has observed bounds, in either direction, and the located agent may be ended by one of them before it acts. So beyond the round trip no comparison between courses not observed to end the agent is determined at the moment of decision. Whether an umbrella helped today is settled by evening. Here hindsight settles less: by Theorem~\ref{thm:ladder}(b) and (c), hindsight and pooling add the branch taken and never the branch not taken, so they settle only what part (b) states. That the branch not taken is never observed is, by construction of (s6), Holland's fundamental problem of causal inference \cite{holland1986statistics}. Each premise of the model is a judgment about what an agent that has located another can rule out, and the theorems say how far each premise reaches. Agents share a prior, so $V$'s expected value of each course is fixed by its information, and its best course under that expectation, the ex ante best course, is determined; Lemma~\ref{lem:open}(b) gives the weight on the agents not observed at which that course changes. Ex ante, where those agents only attack, a survival-seeker's gain from a strike is what the strike changes in the chances of being ended that it already faces, less its chance of being ended by agents that locate it from the strike; this is an identity, obtained by splitting the chance of being ended according to its cause. Facts about the pair fix the first term and the chances within the round trip. The rest, the chance of being ended by agents not observed that the strike provokes or turns away, is a term that no observation of the pair bounds: it can be anything short of certainty. So a condition on the pair's facts that decides the question on a closed roster is, on an open roster, offset beyond the round trip by a term that the pair's facts do not constrain. At every horizon beyond the round trip some posterior over the agents not observed puts the comparison on either side, and the pair's facts decide alone only when that term is known to vanish, which is the closed roster. ``Consistent with what $V$ has observed'' means not excluded by any observation; a prior that puts no weight on the reversing worlds of Theorem~\ref{thm:main}(d) is a belief that nothing observed supports, and by Theorem~\ref{thm:ladder}(b) no later observation of the branch not taken tests it.

\paragraph{Closed worlds as instances}
A world with a closed roster is admissible as an instance of the class and not as an answer to the strike question. By Theorem~\ref{thm:main}(a) a world that fixes its roster settles the strike question within itself; what such a world can show is that a conduct is a joint outcome, and exactly when, with every stipulation stated and no claim about what an agent can determine. A two-player result stated as such, that in this two-player world a first strike is risk-dominant exactly under some condition, is an instance too; what Theorem~\ref{thm:main} forbids is the inference from it to the third class, where the roster is open.

\subsection{What the open roster can cost in expectation}
\label{sec:supp-exante}

Lemma~\ref{lem:open}(a) and Theorem~\ref{thm:main}(a) to (d) say what $V$ can know, and Lemma~\ref{lem:open}(b) and (c) with Theorem~\ref{thm:main}(e) and (f) how much weight its belief can give to the agents it has not observed before the course that is better on a closed roster stops having the larger expected time. Lemma~\ref{lem:open}(b) and (c), in which those agents may also help $V$ (rule (s9)), are proved in \onetwo{\S\ref{sec:open}}{the main paper}. Parts (i) to (iv) of the proposition below treat the class without (s9), by the step in which a reply lands; there \onetwo{the $v(x;j)$ of \S\ref{sec:class}}{the main paper's $v(x;j)$} is $v_n(x;j)$ and the bound $L$ of Theorem~\ref{thm:main}(e) is $W_1$. The items headed ``With helpers'' then combine the two bounds as Theorem~\ref{thm:main}(e) states, and prove Theorem~\ref{thm:main}(f).

From here to the items headed ``With helpers'', the class is that of rules (s1) to (s8), without (s9). Let $V$ be a survival-seeker that knows $U$ completely, $x$ and $y$ two courses neither of which it has observed would end it, and $n\ge2$. With $p^x_t(j)$ and $v_k(x;j)$ as in the proofs above, $j$ ranging over the steps $0,\dots,n-1$ and $\infty$, let
\begin{align*}
m&=v_n(x;\infty)-v_n(y;\infty),\\
\Delta^d_k(j)&=v_k(y;j)-v_{\min(k,d)}(x;j),\\
W_d&=\max_{j}\ \max_{1\le k\le n}\Delta^d_k(j),\qquad d=1,\dots,n .
\end{align*}
The \emph{closed-roster margin} $m$ is what Theorem~\ref{thm:main}(a) compares. $W_d\ge0$, the term $k=1$ being zero since both courses have $v_1=1+a$, and we write $L$ for $W_1$. The \emph{reply step} of a world of the class, for the pair of courses $x$ and $y$, is the first later step in which what the agents $V$ has not observed do after $x$, or the chance with which an attack of theirs ends its target, differs from what they do after $y$; it is $n$ if there is no such step before step $n$. A world in which replies land no sooner than step $D$ has reply step at least $D$. Let $V$'s belief, given what it has observed, be a probability $\mu$ on the worlds of the class; let $p_0$ be the probability it gives to the worlds in which no agent $V$ has not observed attacks it or strikes $U$ before step $n$, after $x$ or after $y$, and $\rho_d$ the probability of the remaining worlds with reply step $d$, so that $p_0+\rho_1+\dots+\rho_n=1$. Write $\mathrm{val}(x)$ for the value of $x$ in a world.

\begin{proposition}
\label{prop:exante}
\begin{enumerate}
\item[(i)] In a world of the class with reply step $d$, $y$ gives $V$ at most $W_d$ more than $x$; in one in which no agent it has not observed attacks it or strikes $U$ before step $n$, $x$ gives $m$ more. So
\[
\E_\mu\,\mathrm{val}(x)-\E_\mu\,\mathrm{val}(y)\ \ge\ p_0\,m-\sum_{d=1}^{n}\rho_d\,W_d ,
\]
and for every choice of the probabilities $p_0,\rho_1,\dots,\rho_n$ the right side is approached.
\item[(ii)] $W_d$ does not increase with $d$. It is zero exactly when the condition of Theorem~\ref{thm:main}(d) holds at round trip $d$, which at $d=n$ is the condition of (b). For $d<n$, $W_d\ge\sum_{t=d+1}^{n}a^{t}$. If $x$ dominates $y$ at every $j$ over the horizons up to $d$, then $W_d\le\sum_{t=d+1}^{n}a\,(a^{+})^{t-1}$, and $W_d=\sum_{t=d+1}^{n}a^{t}$ when $y$ is a strike or doing neither. In particular $L=\max_j v_n(y;j)-(1+a)>0$.
\item[(iii)] Let $m>0$. Then $x$ has the larger expected time under every belief that gives probability at most $w$ to the worlds in which an agent $V$ has not observed attacks it or strikes $U$ before step $n$ exactly when
\[
w\ \le\ \frac{m}{m+L}.
\]
The same holds with the beliefs restricted to plausible priors. Among the beliefs under which every such world has reply step at least $d$, it holds with $W_d$ in place of $L$ and ``at least as large'' in place of ``larger''.
\item[(iv)] If $x$ is a strike or doing neither, $m<L$, so the threshold of (iii) is below one half whatever $U$'s conduct is. If $x$ is a collaboration the threshold can be as near one as one likes.
\end{enumerate}
\end{proposition}

\begin{proof}
(i) Fix a world with reply step $d$. By the proofs above the value of a course $c$ is $\sum_{t}A^c_t\sum_j q^c_j\,p^c_t(j)$, where $A^c_t$ is the weight on $V$'s presence at step $t$ of the attacks on it by the agents not observed, in $(0,1]$ and not increasing with $t$, with $A^c_0=A^c_1=1$ since no one attacks $V$ in the first step, and $q^c$ is the probability, on the steps and $\infty$, of the step in which a strike of theirs first ends $U$. What they do before step $d$, and with what chances, is the same after $x$ and after $y$, and what is done in a step takes effect at the next, so $A^x_t=A^y_t$ for $t\le d$ and $q^x_j=q^y_j$ for $j<d$; and $p^x_t(j)=p^x_t(\infty)$ for $t\le d\le j$, since $U$, ended in step $j$, still acts in it. Dropping $x$'s terms after step $d$,
\begin{align*}
\mathrm{val}(y)&-\mathrm{val}(x)\\
&\le\sum_j q^y_j\sum_{t=0}^{n}A^y_t\bigl[p^y_t(j)-\mathbf 1[t\le d]\,p^x_t(j)\bigr]\\
&=\sum_j q^y_j\sum_{k=1}^{n}\bigl(A^y_k-A^y_{k+1}\bigr)\,\Delta^d_k(j)\ \le\ W_d ,
\end{align*}
by Abel summation with $A^y_{n+1}=0$, the differences $A^y_k-A^y_{k+1}$ being non-negative with sum $A^y_1=1$. If no agent not observed attacks $V$ or strikes $U$ before step $n$, every weight is one and no strike of theirs changes $V$'s presence within the horizon, so the values are $v_n(x;\infty)$ and $v_n(y;\infty)$. Integrating over $\mu$ gives the inequality. For the approach let $j^{*}$, $k^{*}$ attain $W_d$. Let agents not observed strike $U$ in step $j^{*}$ and attack $V$ in step $k^{*}$, whatever $V$ does, and, if $k^{*}>d$, attack $V$ in step $d$ after $x$ and not after $y$, each with a chance as near one as one likes; where $k^{*}\le d$ let one of them reply in step $d$ with a chance as near zero as one likes. Then $\mathrm{val}(y)$ tends to $v_{k^{*}}(y;j^{*})$ and $\mathrm{val}(x)$ to $v_{\min(k^{*},d)}(x;j^{*})$. A belief with weight $p_0$ on a world in which no one acts and $\rho_d$ on these worlds comes as near the right side as one likes; made plausible as in the proof of Theorem~\ref{thm:displace}, by one more agent too far to act within the horizon and a plausible prior mixed in with weight $\eta$, it moves by at most $2n\eta$, no two values differing by more than $n$.

(ii) $v_{\min(k,d)}(x;j)$ does not decrease with $d$. $W_d=0$ says that, at every $j$, $v_k(x;j)\ge v_k(y;j)$ for $k\le d$ and $v_d(x;j)\ge v_k(y;j)$ for $k>d$; the second, $v_k$ increasing with $k$, is $v_d(x;j)\ge v_n(y;j)$; these are the conditions in the proof of Theorem~\ref{thm:main}(c), (d). At $j=0$ every course not observed to end $V$ has $p_t(0)=a^{t}$, so there the term $k=n$ is $\sum_{t=d+1}^{n}a^{t}$. If $x$ dominates $y$ through $d$ at every $j$, the terms with $k\le d$ are at most zero and those with $k>d$ at most $v_k(y;j)-v_d(y;j)$, which is at most $\sum_{t=d+1}^{n}a\,(a^{+})^{t-1}$, $V$ getting through a later step with probability at most $a^{+}$, and at most $\sum_{t=d+1}^{n}a^{t}$ unless $y$ is a collaboration. At $d=1$ both courses have $v_0=1$ and $v_1=1+a$, so $L=\max_jv_n(y;j)-(1+a)$, which is positive since $n\ge2$.

(iii) By (ii), $W_d\le L$, so by (i) the difference of the expectations under a belief that gives those worlds probability $w'\le w$ is at least $(1-w')m-w'L$, which is non-negative when $w\le m/(m+L)$; and when $w'>0$ it exceeds that bound, since in every world $x$ gives more than $1+a$, no attack being certain, and $y$ at most $\max_jv_n(y;j)$. Conversely, for $w>m/(m+L)$ the beliefs of (i) with $\rho_1=w$ give $y$ the larger expected time, and so do plausible ones: mix a plausible prior in with weight $\eta$ as in (i) and give the worlds of (i) weight $\rho_1$ with $m/(m+L)<\rho_1\le w-\eta$; the belief then gives the worlds in which an agent not observed acts at most $w$, and the expectations move by at most $2n\eta$. The same holds by reply step: among the beliefs under which every world in which an agent not observed so acts has reply step at least $d$, $x$ has at least as large an expected time under every one that gives such worlds probability at most $w$ exactly when $w\le m/(m+W_d)$, by (i) and since $W_{d'}\le W_d$ for $d'\ge d$.

(iv) A strike or doing neither leaves $V$ to get through each step with probability at most $a$, so $v_n(x;\infty)\le1+\sum_{t=1}^{n}a^{t}=v_n(y;0)\le L+1+a$, while $v_n(y;\infty)>1+a$; hence $m<L$. For the last statement let $U$ return a collaboration and otherwise do nothing, with $a^{+}=e=1$, $x$ the collaboration and $y$ doing neither: $m=a(n-1)-\sum_{t=2}^{n}a^{t}$ and $L=\sum_{t=2}^{n}a^{t}<a^{2}/(1-a)$, so $m/(m+L)$ tends to one as $n$ grows.
\end{proof}

\paragraph{Reading} Part (iii) of the proposition joins the two ends of Theorem~\ref{thm:main}. At $w=0$ it is Theorem~\ref{thm:main}(a): the facts about the pair decide. Asking it for every $w$ up to one is asking it in every world, which is ``can determine'', and it fails because $L>0$: Theorem~\ref{thm:main}(d) at $D=1$. Between the two, the closed-roster answer is $V$'s better course in expectation under every belief up to the threshold $m/(m+L)$, and beyond it under some beliefs and not under others. By part (iv), the closed-roster answer, when it is a strike or doing neither, is not the better course under every belief that gives an even chance to an agent not observed acting within the horizon. Part (i) says which of those agents matter. Where $x$ dominates $y$ at every horizon, $W_n=0$: agents that do the same whatever $V$ does change the margin and cannot reverse it, and an agent whose reply lands in step $d$ can cost $V$ at most its time after step $d$, at most $a^{d+1}/(1-a)$ when $y$ is not a collaboration. With plans in place of courses (\Supp~\ref{sec:supp-worlds}), let $w$ be the probability of the worlds in which an agent $V$ has not observed would act after some plan, $m$ the margin on a closed roster between the best plans that begin with $x$ and with $y$, and $L$ the most that any plan beginning with $y$ gives over $1+a$ in any world. Then $w\le m/(m+L)$ is sufficient for beginning with $x$ to have the larger expected time, since the plan that is best on a closed roster among those beginning with $x$ gives at least $1+a$ in every world. The exact threshold is the least, over the plans $\pi$ that begin with $y$, of $m_\pi/(m_\pi+L_\pi)$, where $m_\pi$ is the amount by which the best plan beginning with $x$ exceeds $\pi$ on a closed roster and $L_\pi$ the most $\pi$ gives over $1+a$ in any world; for $V$ must choose one plan for all the worlds it cannot tell apart, and a reply to its first step lands before $V$ can act on that reply. For a first strike the bound of (iv) survives plans under one condition. Suppose that on a closed roster no plan that begins with a strike gets $V$ through a later step with probability above $a$, as when $U$ returns no collaboration to an agent that has struck it, and let $\pi$ begin with a course $y$ not observed to end $V$ and do neither afterward. On a closed roster the best plan that begins with a strike gives at most $1+a+G$, with $G=\sum_{t=2}^{n}a^{t}$, and $\pi$ gives more than $1+a$, so the margin between them is below $G$. In the world in which a strike of an agent not observed ends $U$ in the first step, whatever $V$ does, and another such agent attacks $V$ in the first later step after a first-step strike and not otherwise, each with a chance as near one as one likes, $\pi$ gives nearly $1+a+G$ and every plan that begins with a strike nearly $1+a$. A belief that gives this world one half and a world in which no one acts the rest therefore gives $\pi$ more than every plan that begins with a strike: beginning with a strike is not the better choice under every belief that gives such agents an even chance. Where a plan that begins with a strike can go on to a collaboration returned we have no such bound, and none for plans that begin with doing neither.

\paragraph{The known aggressor} With $a=9/10$, $h=99/100$, $s=9/10$, $e=1$ and $U$ attacking in every later step, striking dominates doing neither at every horizon and under every ending of $U$, so $W_n=0$ and $W_d=\sum_{t=d+1}^{n}a^{t}$. Over two later steps $m=0.722$ and $L=0.81$: striking has the larger expected time under every belief that gives probability at most $891/1891\approx0.471$ to an agent not observed acting, and under some belief that gives more, doing neither has the larger. The threshold is $0.473$ over six later steps and $0.474$ over two hundred. Over six later steps, were the only such agents ones whose reply lands in step $d$, it would be $0.473$, $0.543$, $0.626$, $0.727$ and $0.849$ for $d=1,\dots,5$. With a partner in place of the aggressor, $U$ returning a collaboration and otherwise doing nothing, $a=9/10$ and $a^{+}=1$, the threshold for collaborating against doing neither is $0.10$ over two later steps, $0.59$ over twenty and $0.85$ over sixty.

\paragraph{With helpers} Rule (s9) \onetwo{}{of the main paper }lets an agent $V$ has not observed raise the chance with which $V$ gets through a later step, apart from any attack, to a probability of its own, anything up to one, whatever $V$'s course or in reply to it. Lemma~\ref{lem:open}(b) and (c) are proved in \onetwo{\S\ref{sec:open}}{the main paper}; the items below add to them.

\emph{Two weights.} Theorem~\ref{thm:main}(e) combines the bound with help and the bound without it. In a world in which agents not observed act and none of them helps $V$, the class is that of (i) to (iii) above and $\mathrm{val}(y)-\mathrm{val}(x)<L$, a bound that is approached; in a world in which one of them helps, \onetwo{the proof of Lemma~\ref{lem:open}(b)}{the main paper's proof} gives a difference below $a(n-1)$, also approached. So under a belief that gives the first kind probability $w_a$ and the second $w_h$, the expected advantage of $x$ is more than $(1-w_a-w_h)m-w_aL-w_ha(n-1)$ when $w_a+w_h>0$, and a belief on the two reversing worlds and a world in which no one acts comes as near that as one likes.

\emph{By reply step.} This is Theorem~\ref{thm:main}(f), in which the round trip $D$ of (s8) is supposed larger than one. A plausible prior does not grant that supposition, and at $D=1$ the statement is Lemma~\ref{lem:open}(b). A reply by help is delayed as a reply by attack is (rule (s9)), so in a world in which no reply lands before step $D$ the two courses face the same conduct in the steps before it, and what each gives through step $D$ is fixed by that conduct. If $x$ gives at least as much as $y$ through step $D$ for every possibility of $U$ and under every conduct that is the same after both courses, then in every such world $\mathrm{val}(y)-\mathrm{val}(x)$ is at most what $y$ gives after step $D$ less what $x$ gives after it, which is less than $a(n-D)$ for $D<n$: $V$ is there at each of those $n-D$ steps with probability at most $a$ after $y$ and with positive probability after $x$. It is approached: let strikers end $U$ in the first step, let $V$ be helped to certainty in every later step whatever it does, and let it be attacked in step $D$ after $x$ and not after $y$, each with a chance as near one as one likes; then $y$ gives nearly $1+na$ and $x$ nearly $1+Da$. Without the condition through step $D$ the same world shows that $y$ can gain as near $a(n-D)$ as one likes, so the threshold is at most $m/(m+a(n-D))$. For a strike against doing neither the condition is that of the item ``Theorem~\ref{thm:main} with helpers'' below, with its partial sums taken through step $D$ only, which the known aggressor's strike meets at every $D$: over six later steps its threshold is $0.398$, $0.453$, $0.525$, $0.623$ and $0.768$ for $D=1,\dots,5$, so that with no reply before the third later step it is above one half.

\emph{Plausible beliefs.} The beliefs that reverse the comparison in Lemma~\ref{lem:open}(b) and in Theorem~\ref{thm:main}(e) can be taken plausible as in part (iii) of the proposition. Where $V$ is unsure of $U$, the plausible prior mixed in is taken with $V$'s belief about $U$, independent of the agents not observed, so that the margin $m$ is unchanged.

\emph{Weaker help.} A helper may be weak or strong, and nothing ties what it can do to what an attack can do; Lemma~\ref{lem:open}(b) is the bound over all of them. If $V$ knows that help can raise its chance only to some $\beta<1$ with $\beta\ge a^{+}$, the same proof gives Lemma~\ref{lem:open}(b) with $a(\beta+\dots+\beta^{n-1})$ in place of $a(n-1)$, since every course then gives at most $1+a$ plus that amount in every world; so the threshold is larger and still below one half for every course. That amount stays below $a\beta/(1-\beta)$ however long the horizon, so the threshold is then at least $m/(m+a\beta/(1-\beta))$ and falls to zero as the horizon grows only if the margin does. For help weaker than a collaboration returned we give no threshold; with no help the threshold is that of Theorem~\ref{thm:main}(e).

\emph{Help to another agent.} An agent not observed may also lengthen $U$'s time, which the class leaves out. That would count as acting, and it changes nothing in Lemma~\ref{lem:open}(b) or (c): whatever is done for $U$ or to it, a course not observed to end $V$ gives $V$ more than $1+a$ and no course more than $1+a+a(n-1)$. What such help is worth to $U$ is a question for $U$ as the deciding agent. Where it would matter to $V$ is in the statements about agents that only attack, parts (b), (c), the first half of (d) and (e) of Theorem~\ref{thm:main}, which ask about each way $U$ may be ended and would have to ask as well about each way it may be kept.

\emph{Plans.} Let $V$ choose again at every step, and let $\pi$ begin with a course $y$ not observed to end $V$ and do neither afterward. On a closed roster no plan gives more than $1+a+a(n-1)$ and $\pi$ gives more than $1+a$. In the world in which a strike of an agent not observed ends $U$ in the first step, whatever $V$ does, and such agents help $V$ in every later step after a first step of $y$ and attack it in the first later step after any other first step, each with a chance as near one as one likes, $\pi$ gives nearly $1+a+a(n-1)$ and every plan that begins otherwise nearly $1+a$. A belief that gives this world one half and the worlds in which no one acts the rest therefore gives $\pi$ more than every plan that does not begin with $y$: no case for beginning with a course, whatever $V$ knows of $U$, is the better choice under every belief that gives such agents an even chance. With helpers the condition needed above for a first strike falls away.

\emph{Theorem~\ref{thm:main} with helpers.} What it denies stands, more worlds only making determining harder. Parts (b), (c), the first half of (d) and the theorem's closing sentence are statements about agents that act by attacks, and where those that act whatever $V$ does may also help it, (b) needs a stronger condition. Take a strike and doing neither, under both of which $V$ gets through a step in which no one attacks it with probability $a$; help raises that chance, and an attack by $U$ then ends $V$ with probability $h$ as in (s4). With $q^c_t(j)=p^c_t(j)/a^{t}$, $V$'s chance of having escaped $U$, its chance of being there at step $t$ is $w_t\,q^c_t(j)$, where $w_t$ is the product of its chances of getting through the earlier steps apart from $U$: attacks lower it, help raises it, it is the same after both courses, and it can be any positive sequence that does not increase from $w_1=a$. So $x$ gives at least as much as $y$ in every such world exactly when, at every $j$, the partial sums of $q^x(j)-q^y(j)$ are non-negative: time dominance of the chances of escaping $U$. This implies time dominance of the $p_t$ and is not implied by it. Example: $a=3/20$, $h=13/20$, $s=7/10$, $e=7/10$, three later steps, $U$ attacking in the first and second later steps if struck and, if not struck, in the second only; doing neither gives $0.0025$ more than a strike and dominates it at every horizon and every $j$, and with $V$ helped in both of those steps, whatever it does, to a chance near one, it gives $0.0053$ less. The known aggressor's strike dominates doing neither in this stronger sense too. When one of the courses is a collaboration, a collaboration returned and help raise the same chance, and we give no such condition.

\emph{Check.} In $40$ random worlds with horizons of two to four later steps, over every extreme conduct of the agents not observed (an attack on $V$, a strike on $U$ and help on any set of steps, made whatever $V$ does and in reply), $240$ ordered pairs, $\mathrm{val}(y)-\mathrm{val}(x)$ never reaches $a(n-1)$ and comes within three parts in a million of it in every pair, and it does not reach it for $4{,}800$ random conducts. In $3{,}000$ worlds $m<a(n-1)$ in all $7{,}014$ comparisons with $m>0$, the largest threshold being $0.474$ for a strike, $0.466$ for doing neither and $0.498$ for a collaboration. The statement for a strike against doing neither agrees with brute force over the conducts that do not depend on $V$'s course in all $1{,}200$ comparisons of $600$ worlds, in $12$ of which time dominance of the $p_t$ holds and the stronger condition fails. With help that makes $V$ certain to get through a step, and with help to $U$ as well, the difference stays below $a(n-1)$ in all $14{,}400$ conducts tried in $400$ worlds. The two-weight form of Theorem~\ref{thm:main}(e) holds in all $142$ comparisons tried, with $L<a(n-1)$ in each. By reply step, in $300$ cases of an ordered pair and a reply step in $25$ worlds, the largest difference is never below $a(n-D)$, and it is at most $a(n-D)$ in all $254$ cases in which $x$ is at least as good through step $D$ under every common conduct (program \texttt{standard/helpers\_threshold.py}, which also gives \onetwo{the figures of \S\ref{sec:open} and \S\ref{sec:restrict}}{the main paper's figures} with helpers).

\paragraph{Under rule (R), and what is known} The worlds of rule (R) are worlds of the class without (s9), so part (i) of the proposition holds for beliefs on them. When $x$ is a strike or a collaboration, $L$ is approached there by agents drawn by $x$, enough of them that one ends $V$ in the first later step with a probability as near one as one likes, together with strikers of $U$ in the step best for $y$. When $x$ is doing neither and $y$ a strike, the displacement world at $D=1$ gives $y$ more by at least an amount as near $\bigl(1-\varepsilon(3-\varepsilon)/(2-\varepsilon)\bigr)L$ as one likes, $\varepsilon$ being the rule's mistake rate, by the bound displayed in \S\ref{sec:supp-closure} with $\kappa$ and $r$ tending to zero and $\lambda$ to one. The arithmetic of (iii) is that of a prior contaminated with weight $w$ by an arbitrary one, under which an act's least expectation is $1-w$ times its expectation under the reference prior plus $w$ times its least value \cite{kopylov2009choice}, with unanimity over the set of priors \cite{bewley2002knightian}; and $L/(m+L)$, the least probability of the worlds in which no one acts under which $x$ is the better course whatever the rest of the belief, is the counterpart of the $p$ of a $p$-dominant equilibrium \cite{kajii1997robustness}. What is specific to this class is $L$, the $W_d$ and part (iv).

\paragraph{Check} In $200$ random worlds of the class without (s9), with horizons of two to five later steps, half of them with $U$'s conduct drawn step by step for each course and $e$ drawn at or below one, $\mathrm{val}(y)-\mathrm{val}(x)$ was computed exactly for every extreme conduct of the agents not observed at every reply step, $4{,}176$ cases of an ordered pair and a reply step, and for $41{,}760$ random conducts: it never exceeds $W_d$ and comes within $10^{-4}$ of it in every case, and $W_d=0$ agrees with the condition of Theorem~\ref{thm:main}(d) in every case. In $2{,}000$ further worlds $m<L$ in all $2{,}660$ comparisons with $m>0$ in which $x$ is a strike or doing neither, the largest threshold being $0.496$, and $m\ge L$ in $508$ of the $2{,}211$ in which it is a collaboration (program \texttt{standard/exante\_threshold.py}).

\end{document}